\documentclass[]{article}
\usepackage[english]{babel}
\usepackage{tikz}
\usepackage{amssymb}
\usepackage[T1]{fontenc}
\usepackage[utf8]{inputenc}
\usepackage{amsmath}
\usepackage{amsthm}
\usepackage{enumerate}
\usepackage{caption}
\usepackage{graphicx}
\usepackage{authblk}
\usetikzlibrary[shapes.misc]
\usetikzlibrary[shapes.geometric]
\newtheorem{de}{Definition}[section]
\newtheorem{theo}{Theorem}[section]
\newtheorem{prop}[theo]{Proposition}
\newtheorem{cor}[theo]{Corollary}
\newtheorem{lem}[theo]{Lemma}
\newtheorem{obs}[theo]{Observation}

\title{A characterization of b-chromatic and partial Grundy numbers by induced subgraphs}

\author[2]{Brice Effantin}
\author[1]{Nicolas Gastineau}
\author[1]{Olivier Togni}
\affil[1]{LE2I UMR6306, CNRS, Arts et Métiers, \textit{Univ. Bourgogne Franche-Comté, F-21000 Dijon, France} }
\affil[2]{\textit{Université de Lyon, CNRS, Université Lyon 1, LIRIS, UMR5205, F-69622, France} }

\begin{document}
\maketitle
\begin{abstract}
Gyárfás et al. and Zaker have proven that the Grundy number of a graph $G$ satisfies $\Gamma(G)\ge t$ if and only if $G$ contains an induced subgraph called a $t$-atom.
The family of $t$-atoms has bounded order and contains a finite number of graphs.
In this article, we introduce equivalents of $t$-atoms for b-coloring and partial Grundy coloring.
This concept is used to prove that determining if $\varphi(G)\ge t$ and $\partial\Gamma(G)\ge t$ (under conditions for the b-coloring), for a graph $G$, is in XP with parameter $t$.
We illustrate the utility of the concept of $t$-atoms by giving results on b-critical vertices and edges, on b-perfect graphs and on graphs of girth at least $7$.
\end{abstract}

\section{Introduction}
Given a graph $G$, a \emph{proper $k$-coloring} of $G$ is a surjective function $c:V(G) \rightarrow\{1,\ldots,k\}$ such that $c(u)\neq c(v)$ for any $uv\in E(G)$; the \emph{color class} $V_i$ is the set $\{u\in V| c(u)=i\}$ and a vertex $v$ has \emph{color} $i$ if $v\in V_i$. We denote by $N(u)$ the set of neighbors of a vertex $u$ and by $N[u]$ the set $N(u)\cup\{u\}$.
A vertex $v$ of color $i$ is a \emph{Grundy vertex} if it is adjacent to at least one vertex colored $j$, for every $j<i$.
A \emph{Grundy $k$-coloring} is a proper $k$-coloring such that every vertex is a Grundy vertex. 
The \emph{Grundy number} of a graph $G$, denoted by $ \Gamma(G)$, is the largest integer $k$ such that there exists a Grundy $k$-coloring of $G$ \cite{GR1939}.
A \emph{partial Grundy $k$-coloring} is a proper $k$-coloring such that every color class contains at least one Grundy vertex.
The \emph{partial Grundy number} of a graph $G$, denoted by $\partial \Gamma(G)$, is the largest integer $k$ such that there exists a partial Grundy $k$-coloring of $G$.
Let $G$ and $G'$ be two graphs. By $G\cup G'$ we denote the graph with vertex set $V(G)\cup V(G')$ and edge set $E(G)\cup E(G')$.
Let $m(G)$ be the largest integer $m$ such that $G$ has at least $m$ vertices of degree at least $m-1$.
A graph $G$ is \emph{tight} if it has exactly $m(G)$ vertices of degree $m(G)-1$.

Another coloring parameter with domination constraints on the colors is the \emph{b-chromatic number}.
In a proper-$k$-coloring, a vertex $v$ of color $i$ is a \emph{b-vertex} if $v$ is adjacent to at least one vertex colored $j$, $1\le j\neq i\le k$.
A \emph{b-$k$-coloring}, also called b-coloring when $k$ is not specified, is a proper $k$-coloring such that every color class contains at least one b-vertex.
The \emph{b-chromatic number} of a graph $G$, denoted by $\varphi(G)$, is the largest integer $k$ such that there exists a b-$k$-coloring of $G$.
In this paper, we introduce the concept of \emph{b-relaxed number}, denoted by $\varphi_r(G)$.
A \emph{b-$k$-relaxed coloring} of $G$ is a b-$k$-coloring of a subgraph of $G$.
The b-relaxed number of $G$ is $\varphi_r(G)=\max_{H\subseteq G}( \varphi(H))$, for $H$ an induced subgraph of $G$.
Note that we have $\varphi(G)\le \varphi_r(G) \le \partial \Gamma(G)$. The difference between $\varphi(G)$ and $\varphi_r(G)$ can be arbitrary large. Let $K^{-}_{n,n}$ denotes the complete bipartite graph $K_{n,n}$ in which we remove $n-1$ pairwise non incident edges (or $n-1$ edges of a perfect matching in $K_{n,n}$) \cite{BA2013}. For this graph we have $\varphi(K^{-}_{n,n})=2$ and $\varphi_r(K^{-}_{n,n})=n$ as Figure \ref{b-diff} illustrates it (for $n=3$).

\begin{figure}[t]
\begin{center}
\begin{tikzpicture}

\draw (0,0) -- (1,1);
\draw (0,0) -- (2,1);
\draw (1,0) -- (0,1);
\draw (1,0) -- (2,1);
\draw (2,0) -- (0,1);
\draw (2,0) -- (1,1);
\draw (2,0) -- (2,1);
\node at (0,0) [circle,draw=black,fill=black, scale=0.7] {};
\node at (1,0) [circle,draw=black,fill=black, scale=0.7] {};
\node at (2,0) [circle,draw=black,fill=black, scale=0.7] {};
\node at (0,1) [circle,draw=red,fill=red, scale=0.7] {};
\node at (1,1) [circle,draw=red,fill=red, scale=0.7] {};
\node at (2,1) [circle,draw=red,fill=red, scale=0.7] {};
\node at (0,0-0.3) {1};
\node at (1,0-0.3) {1};
\node at (2,0-0.3) {1};
\node at (0,1+0.3) {2};
\node at (1,1+0.3) {2};
\node at (2,1+0.3) {2};

\draw (0+5,0) -- (1+5,1);
\draw (0+5,0) -- (2+5,1);
\draw (1+5,0) -- (0+5,1);
\draw (1+5,0) -- (2+5,1);
\draw (2+5,0) -- (0+5,1);
\draw (2+5,0) -- (1+5,1);
\draw (2+5,0) -- (2+5,1);
\node at (0+5,0) [circle,draw=black,fill=black, scale=0.7] {};
\node at (0+5,1) [circle,draw=black,fill=black, scale=0.7] {};
\node at (1+5,0) [circle,draw=red,fill=red, scale=0.7] {};
\node at (1+5,1) [circle,draw=red,fill=red, scale=0.7] {};
\node at (2+5,0) [circle,draw=blue,fill=blue, scale=0.7] {};
\node at (2+5,1) [circle,draw=black,fill=black, scale=0.5] {};
\node at (0+5,0-0.3) {1};
\node at (1+5,0-0.3) {2};
\node at (2+5,0-0.3) {3};
\node at (0+5,1+0.3) {1};
\node at (1+5,1+0.3) {2};

\end{tikzpicture}
\end{center}
\caption{The graph $K^{-}_{3,3}$ with $\varphi(K^{-}_{3,3})=2$ (on the left) and $\varphi_r (K^{-}_{3,3})=3$ (on the right).}
\label{b-diff}
\end{figure}
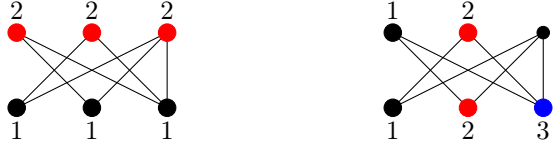

The concept of b-coloring has been introduced by Irving and Manlove \cite{IR1999}, and a large number of papers was published (see e.g. \cite{EF2003,MA2013}).
The b-chromatic number of regular graphs has been investigated in a serie of papers (\cite{CA2011,KL2010,EL2009,SH2012}).
Determining the b-chromatic number of a tight graph is NP-hard even for a connected bipartite graph \cite{KR2002} and a tight chordal graph \cite{HA2012}.

In this paper, we study the decision problems b-COL, b-r-COL and pG-COL with parameter $t$ from Table \ref{tab1}.
\begin{center}
\begin{table}[!ht]
\centering\begin{tabular}{|c|c|c|c|c|}
\hline
 & b-COL & b-r-COL & G-COL& pG-COL \\\hline
Question & Does $\varphi(G)\ge t$? & Does $\varphi_r(G)\ge t$? & Does $ \Gamma(G)\ge t$? & Does $\partial \Gamma(G)\ge t$?  \\ \hline
Complexity &  undetermined &  XP &  XP \cite{ZA2006} & XP \\
class &  & & &  \\ \hline
\end{tabular}
\caption{\label{tab1} The different decision problems with input a graph $G$ and parameter $t$ and their complexity class.}
\end{table}
\end{center}
A decision problem is in FPT with parameter $t$ if there exists an algorithm which resolves the problem in time $O(f(t)\ n^{c} )$, for an instance of size $n$, a computable function $f$ and a constant $c$.
A decision problem is in XP with parameter $t$ if there exists an algorithm which resolves the problem in time $O(f(t)\ n^{g(t)} )$, for an instance of size $n$ and two computable functions $f$ and $g$.

The concept of $t$-atom was introduced independently by Gyárfás et al. \cite{GAR1997} and by Zaker \cite{ZA2006}. The family of $t$-atoms is finite and the presence of a $t$-atom can be determined in polynomial time for a fixed $t$. The following definition is slightly different from the definitions of Gyárfás et al. or Zaker, insisting more on the construction of every $t$-atom (some $t$-atoms can not be obtained with the initial construction of Zaker).
\begin{de}[\cite{ZA2006}]
The family of $t$-atoms is denoted by $\mathcal{A}^{\Gamma}_t$, for $t\ge 1$, and is defined by induction. The family $\mathcal{A}^{\Gamma}_1$ only contains $K_1$.
A graph $G$ is in $\mathcal{A}^{\Gamma}_{t+1}$ if there exists a graph $G'$ in $\mathcal{A}^{\Gamma}_t$ and an integer $m$, $m\le|V(G')|$, such that $G$ is composed of $G'$ and an independent set $I_m$ of order $m$, adding edges between $G'$ and $I_m$ such that every vertex in $G'$ is adjacent to at least one vertex in $I_m$.
\end{de}
Moreover, in the following sections, we say that a graph $G$ in a family of graphs $\mathcal{F}$ is \emph{minimal}, if no graphs of $\mathcal{F}$ is a proper induced subgraph of $G$. For example, a minimal $t$-atom $A$ is a $t$-atom for which there are no $t$-atoms which are induced in $A$ other than itself. 
\begin{theo}[\cite{GAR1997,ZA2006}] 
A graph $G$ satisfies $\Gamma(G)\ge t$ if and only if it contains an induced minimal $t$-atom.
\end{theo}

In this paper we prove equivalent theorems for b-relaxed number and partial Grundy number. In contrast with the minimal $t$-atoms, we can not define the minimal $t$-atoms for b-coloring as the smallest graphs such that $G$ satisfies $\varphi(G)=t$ (also called b-critical graphs). 

The paper is organized as follows: Section 2 is devoted to the definition of $t$-atoms for the partial Grundy coloring. This concept allows us to prove that the partial Grundy coloring problem is in XP with parameter $t$. Section 3 is similar to Section 2 but for $b$-relaxed-coloring. Section 4 is devoted to the concept of b-critical vertices and edges. Section 5 is about b-perfect graphs. Finally, Section 6 deals with graphs for which the b-relaxed and the b-chromatic numbers are equal. 
\section{Partial-Grundy-$t$-atoms: $t$-atoms for partial Grundy coloring}
We start this section with the definition of $t$-atoms for partial Grundy coloring.
\begin{de}
Given an integer $t$, a partial Grundy $t$-atom (or pG-$t$-atom, for short) is a graph $A$ whose
vertex-set can be partitioned into $t$ sets $D_1$, $\ldots$, $D_t$, where $D_i$ contains a special vertex $c_i$ for each $i\in \{1,\ldots,t\}$ such that the following holds:
\begin{itemize}
\item For all $i\in \{1,\ldots,t\}$, $D_i$ is an independent set and $|D_i|\le t-i+1$;
\item For all $i\in \{2,\ldots,t\}$, $c_i$ has a neighbor in each of $D_1$, $\ldots$, $D_{i-1}$.
\end{itemize}
The set $\{c_1 , \ldots,c_t\}$ is called the \emph{center} of $A$ and denoted by $C(A)$.
\end{de}

Note that the sets $D_1$, $\ldots$, $D_t$ induce a partial Grundy coloring of the pG-$t$-atom.
Figure \ref{pgatom} illustrates several pG-$t$-atoms (and their induced colorings) obtained using the previous definition.
\begin{obs}
For every pG-$t$-atom $G$, we have $|V(G)|\le \frac{t(t+1)}{2}$.
\end{obs}
\begin{lem}\label{recptat}
Let $t$ and $t'$ be two integers such that $1\le t' <t$. Every pG-$t$-atom contains a pG-$t'$-atom as induced subgraph.
\end{lem}
\begin{proof}
Every pG-$t$-atom $G$ contains a pG-$t'$-atom $G'$: we can obtain $G'$ by removing every vertex in $D_k$, for $t'<k\le t$, and by removing, afterwards, the vertices of $G'$ not adjacent to any vertex in $\{c_1,\ldots ,c_{t'}\}$.
\end{proof}
Note that the only minimal pG-$2$-atom is $P_2$.
The minimal pG-$3$-atoms are $C_3$, $P_4$ and $P_2\cup P_3 $.
These graphs are illustrated in Figure \ref{pgatom}.
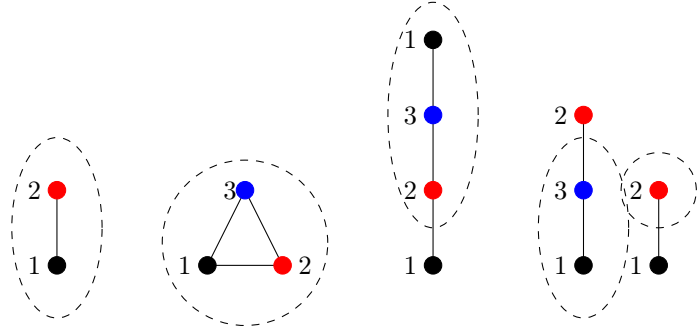
\begin{figure}[t]
\begin{center}
\begin{tikzpicture}
\draw (0,0) -- (0,1);
\node at (0,0) [circle,draw=black,fill=black, scale=0.7] {};
\node at (0,1) [circle,draw=red,fill=red, scale=0.7] {};
\node at (-0.3,0) {1};
\node at (-0.3,1) {2};
\draw[dashed] (0,0.5) ellipse (0.6cm and 1.2cm);

\draw (2,0) -- (3,0);
\draw (2,0) -- (2.5,1);
\draw (3,0) -- (2.5,1);
\node at (2,0) [circle,draw=black,fill=black, scale=0.7] {};
\node at (3,0) [circle,draw=red,fill=red, scale=0.7] {};
\node at (2.5,1) [circle,draw=blue,fill=blue, scale=0.7] {};
\node at (1.7,0) {1};
\node at (3.3,0) {2};
\node at (2.3,1) {3};
\draw[dashed] (2.5,0.3) circle (1.1cm);

\draw (5,0) -- (5,3);
\node at (5,0) [circle,draw=black,fill=black, scale=0.7] {};
\node at (5,1) [circle,draw=red,fill=red, scale=0.7] {};
\node at (5,2) [circle,draw=blue,fill=blue, scale=0.7] {};
\node at (5,3) [circle,draw=black,fill=black, scale=0.7] {};
\node at (4.7,0) {1};
\node at (4.7,1) {2};
\node at (4.7,2) {3};
\node at (4.7,3) {1};
\draw[dashed] (5,2) ellipse (0.6cm and 1.5cm);

\draw (7,0) -- (7,2);
\draw (8,0) -- (8,1);
\node at (7,0) [circle,draw=black,fill=black, scale=0.7] {};
\node at (7,2) [circle,draw=red,fill=red, scale=0.7] {};
\node at (7,1) [circle,draw=blue,fill=blue, scale=0.7] {};
\node at (8,0) [circle,draw=black,fill=black, scale=0.7] {};
\node at (8,1) [circle,draw=red,fill=red, scale=0.7] {};
\node at (6.7,0) {1};
\node at (6.7,1) {3};
\node at (6.7,2) {2};
\node at (7.7,0) {1};
\node at (7.7,1) {2};

\draw[dashed] (7,0.5) ellipse (0.6cm and 1.2cm);
\draw[dashed] (8,1) circle (0.5cm);
\end{tikzpicture}
\end{center}
\caption{The minimal pG-$2$-atom (on the left) and the three minimal pG-$3$-atoms (the numbers are the colors of the vertices and the surrounded vertices form the centers).}
\label{pgatom}
\end{figure}
\begin{theo}\label{psensd}
For a graph $G$, we have $\partial \Gamma(G)\ge t$ if and only if $G$ contains an induced minimal pG-$t$-atom.
\end{theo}
\begin{proof}
Suppose that $\partial \Gamma(G)=t'$ with $t'\ge t$. By definition, there exists a partial Grundy coloring of $G$ with $t'$ colors. Let $u_1$, $\ldots$, $u_{t'}$ be a set of Grundy vertices, each in a different color class of $V(G)$.
The graph induced by $N[u_1] \cup \ldots \cup N[u_{t'}]$ contains a pG-$t'$-atom. Hence, by Lemma \ref{recptat}, since $G$ contains an induced pG-$t'$-atom, then it also contains an induced minimal pG-$t$-atom.

Suppose $G$ contains an induced minimal pG-$t$-atom. Thus, the sets $D_1$, \ldots, $D_t$ induce a partial-Grundy coloring of this pG-$t$-atom.
We can extend this coloring to a partial Grundy coloring of $G$ with at least $t$ colors in a greedy way by coloring the remaining vertices in any order, assigning to each of them the smallest color not used by its neighbors.
\end{proof}

\begin{prop}
Let $G$ be a graph of order $n$ and let $t$ be an integer.
There exists an algorithm in time $O(n^{\frac{t(t+1)}{2}})$ to determine if $\partial \Gamma(G)\ge t$. Hence, the problem pG-COL with parameter $t$ is in XP.
\end{prop}
\begin{proof}
By Theorem \ref{psensd}, it suffices to verify that $G$ contains an induced minimal pG-$t$-atom to have $\partial \Gamma(G)\ge t$.
Since the order of a minimal pG-$t$-atom is bounded by $\frac{t(t+1)}{2}$, we obtain an algorithm in time $O(n^{\frac{t(t+1)}{2}})$.
\end{proof}
We finish this section by determining every graph $G$ with $\partial \Gamma(G)=2$.
\begin{prop}
For a graph $G$ without isolated vertices, we have $\partial \Gamma(G)=2$ if and only if $G=K_{n,m}$, for $n\ge 2$ and $m\ge1$ or $G$ only contains isolated edges.
\end{prop}
\begin{proof}
Zaker \cite{ZA2006} has proven that $\Gamma(G)=2$ if and only if $G$ is the disjoint union of copies of some $K_{n,m}$, for $n\ge 1$ and $m\ge1$. Let $n$ and $m$ be positive integers. We can note that a graph containing a copy of $K_{n,m}$, for $n\ge 2$ and $m\ge1$ and a copy of $K_{n,m}$, for $n\ge 1$ and $m\ge1$ contains an induced $P_3\cup P_2$, hence a pG-$3$-atom. Hence, if $\partial \Gamma(G)=2$, then $G=K_{n,m}$, for $n\ge 2$ and $m\ge1$ or $G$ only contains isolated edges.

 Moreover, neither $K_{n,m}$ nor $P_2\cup\ldots\cup P_2$ does contain an induced $C_3$, $P_4$ or $P_3\cup P_2$. Hence, $\partial \Gamma(K_{n,m})=2$.

\end{proof}
\section{b-$t$-atoms: $t$-atoms for b-coloring}
As in the previous section, we start this section with the definition of b-$t$-atoms (the notion of $t$-atom for b-coloring).
\begin{de}
Given an integer $t$, a b-$t$-atom is a graph A whose vertex-set can be partitioned into $t$ sets $D_1$, \ldots, $D_t$, where $D_i$ contains a special vertex $c_i$ for each $i\in\{1,\ldots,t\}$ such that the following holds:
\begin{itemize}
\item For each $i \in \{1,\ldots, t\}$, $D_i$ is an independent set and $|D_i|\le t$;
\item For all $i, j\in \{1,\ldots t\}$, with $i \neq j$, $c_i$ has a neighbor in $D_j$.
\end{itemize}
The set $\{c_1, \ldots, c_t\}$ is called the center of $A$ and denoted by $C(A)$.
\end{de}

Note that the sets $D_1$, $\ldots$, $D_t$ induce a b-coloring of the b-$t$-atom.
Figure \ref{batom} illustrates several b-$t$-atoms (and their induced coloring) obtained using the previous definition.
\begin{obs}\label{taillbatom}
For every b-$t$-atom $G$, we have $|V(G)|\le t^2$.
\end{obs}
\begin{lem}\label{recbtat}
Let $t$ and $t'$ be two integers such that $1\le t'<t$. Every b-$t$-atom contains a b-$t'$-atom as induced subgraph.
\end{lem}
\begin{proof}
Every b-$t$-atom $G$ contains a b-$t'$-atom $G'$: we can obtain $G'$ by removing every vertex in $D_k$, for $t'<k\le t$, and by removing, afterwards, the vertices not adjacent to any vertex in $\{c_1,\ldots ,c_{t'}\}$.
\end{proof}
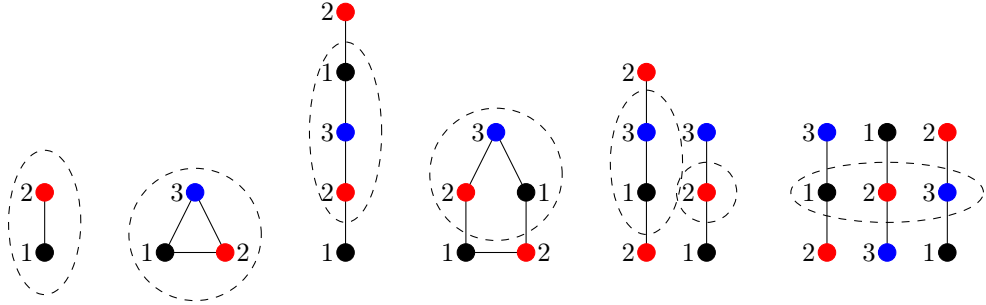
\begin{figure}[t]
\begin{center}
\begin{tikzpicture}[scale=0.8]
\draw (0,0) -- (0,1);
\node at (0,0) [circle,draw=black,fill=black, scale=0.7] {};
\node at (0,1) [circle,draw=red,fill=red, scale=0.7] {};
\node at (-0.3,0) {1};
\node at (-0.3,1) {2};
\draw[dashed] (0,0.5) ellipse (0.6cm and 1.2cm);

\draw (2,0) -- (3,0);
\draw (2,0) -- (2.5,1);
\draw (3,0) -- (2.5,1);
\node at (2,0) [circle,draw=black,fill=black, scale=0.7] {};
\node at (3,0) [circle,draw=red,fill=red, scale=0.7] {};
\node at (2.5,1) [circle,draw=blue,fill=blue, scale=0.7] {};
\node at (1.7,0) {1};
\node at (3.3,0) {2};
\node at (2.2,1) {3};
\draw[dashed] (2.5,0.3) circle (1.1cm);

\draw (5,0) -- (5,4);
\node at (5,0) [circle,draw=black,fill=black, scale=0.7] {};
\node at (5,1) [circle,draw=red,fill=red, scale=0.7] {};
\node at (5,2) [circle,draw=blue,fill=blue, scale=0.7] {};
\node at (5,3) [circle,draw=black,fill=black, scale=0.7] {};
\node at (5,4) [circle,draw=red,fill=red, scale=0.7] {};
\node at (4.7,0) {1};
\node at (4.7,1) {2};
\node at (4.7,2) {3};
\node at (4.7,3) {1};
\node at (4.7,4) {2};
\draw[dashed] (5,2) ellipse (0.6cm and 1.5cm);

\draw (7,0) -- (7,1);
\draw (8,0) -- (8,1);
\draw (7,0) -- (8,0);
\draw (7,1) -- (7.5,2);
\draw (8,1) -- (7.5,2);
\node at (7,0) [circle,draw=black,fill=black, scale=0.7] {};
\node at (8,0) [circle,draw=red,fill=red, scale=0.7] {};
\node at (7.5,2) [circle,draw=blue,fill=blue, scale=0.7] {};
\node at (7,1) [circle,draw=red,fill=red, scale=0.7] {};
\node at (8,1) [circle,draw=black,fill=black, scale=0.7] {};
\node at (6.7,0) {1};
\node at (6.7,1) {2};
\node at (7.2,2) {3};
\node at (8.3,0) {2};
\node at (8.3,1) {1};
\draw[dashed] (7.5,1.3) circle (1.1cm);

\draw (10,0) -- (10,3);
\draw (11,0) -- (11,2);
\node at (10,0) [circle,draw=red,fill=red, scale=0.7] {};
\node at (10,2) [circle,draw=blue,fill=blue, scale=0.7] {};
\node at (10,1) [circle,draw=black,fill=black, scale=0.7] {};
\node at (10,3) [circle,draw=red,fill=red, scale=0.7] {};
\node at (11,0) [circle,draw=black,fill=black, scale=0.7] {};
\node at (11,1) [circle,draw=red,fill=red, scale=0.7] {};
\node at (11,2) [circle,draw=blue,fill=blue, scale=0.7] {};
\node at (9.7,0) {2};
\node at (9.7,1) {1};
\node at (9.7,2) {3};
\node at (9.7,3) {2};
\node at (10.7,0) {1};
\node at (10.7,1) {2};
\node at (10.7,2) {3};

\draw[dashed] (10,1.5) ellipse (0.6cm and 1.2cm);
\draw[dashed] (11,1) circle (0.5cm);

\draw (13,0) -- (13,2);
\draw (14,0) -- (14,2);
\draw (15,0) -- (15,2);
\node at (13,0) [circle,draw=red,fill=red, scale=0.7] {};
\node at (13,2) [circle,draw=blue,fill=blue, scale=0.7] {};
\node at (13,1) [circle,draw=black,fill=black, scale=0.7] {};
\node at (14,1) [circle,draw=red,fill=red, scale=0.7] {};
\node at (14,0) [circle,draw=blue,fill=blue, scale=0.7] {};
\node at (14,2) [circle,draw=black,fill=black, scale=0.7] {};
\node at (15,2) [circle,draw=red,fill=red, scale=0.7] {};
\node at (15,1) [circle,draw=blue,fill=blue, scale=0.7] {};
\node at (15,0) [circle,draw=black,fill=black, scale=0.7] {};
\node at (12.7,0) {2};
\node at (12.7,1) {1};
\node at (12.7,2) {3};
\node at (13.7,2) {1};
\node at (13.7,1) {2};
\node at (13.7,0) {3};
\node at (14.7,0) {1};
\node at (14.7,1) {3};
\node at (14.7,2) {2};
\draw[dashed] (14,1) ellipse (1.6cm and 0.5cm);
\end{tikzpicture}
\end{center}
\caption{The minimal b-$2$-atom (on the left) and the five minimal b-$3$-atoms.}
\label{batom}
\end{figure}

Note that the only minimal b-$2$-atom is $P_2$.
The minimal b-$3$-atoms are $C_3$, $P_5$, $C_5$, $P_3\cup P_4 $ and $P_3\cup P_3 \cup P_3$.
These graphs are illustrated in Figure \ref{batom}.
\begin{obs}
Every minimal pG-$t$-atom is an induced subgraph of a minimal b-$t$-atom or a minimal $t$-atom (an atom for the Grundy number).
\end{obs}

\begin{prop}\label{bsensd}
Let $G$ be a graph.
If $\varphi(G)\ge t$, then $G$ contains an induced minimal b-$t$-atom.
\end{prop}
\begin{proof}
Suppose that $\varphi(G)=t'$, with $t'\ge t$. Thus, there exists a b-coloring of $G$ with $t'$ colors. Let $u_1$, $\ldots$, $u_{t' }$ be a set of b-vertices, each in a different color class of $V(G)$.
The graph induced by $N[u_1] \cup \ldots \cup N[u_{t'}]$ contains a b-$t'$-atom. Hence, by Lemma \ref{recbtat}, since $G$ contains an induced b-$t'$-atom, then it also contains an induced minimal b-$t$-atom.
\end{proof}
\begin{theo}\label{bsensd2}
For a graph $G$, we have $\varphi_r (G)\ge t$ if and only if $G$ contains an induced minimal b-$t$-atom.
\end{theo}
\begin{proof}
Suppose that the graph $G$ contains an induced b-$t$-atom $A$. Since $A$ admits, by definition, a b-$t$-coloring, we have $\varphi_r (G)\ge t$. Using Proposition \ref{bsensd}, we obtain the converse.
\end{proof}
\begin{de}
Let $G$ be a graph.
For an induced subgraph $A$ of $G$, let $N(A)=\{v\in V(G) \setminus V(A) |\ uv \in E(G),\ u\in V(A)\}$.
A b-$t$-atom $A$ is \emph{feasible} in $G$ if there exists a b-$t$-coloring of $V(A)$ that can be extended to the vertices of $N(A)$ without using new colors.
\end{de}
\begin{prop}\label{bsensi}
Let $G$ be a graph.
If $G$ contains an induced feasible minimal b-$t$-atom and no induced feasible minimal b-$t'$-atom, for $t'>t$, then $\varphi(G)=t$.
\end{prop}
\begin{proof}
Suppose that $G$ contains an induced feasible minimal b-$t$-atom $A$ and no b-$t$-coloring of $G$ exists.
We begin by considering that the vertices of $A\cup N(A)$ are already colored with $t$ colors.
We can note that, by assumption, no coloring of $A\cup N(A)$ (from the definition) can be extended to the whole graph using only $t$ colors.
Let $t'$ be the largest integer such that the coloring can not be extended to a b-$t'$-coloring of the whole graph and let $v$ be a vertex that can not be given a color among $\{1,\ldots, t'\}$. Thus, we suppose that the coloring can be extended to a b-$(t'+1)$-coloring where $v$ is colored by $t'+1$.
Since $A\cup N(A)$ is already colored, we have $v\in V(G)\setminus (A\cup N(A))$.
The vertex $v$ should be adjacent to vertices of every color, otherwise it could be colored.
One vertex of each color class in $N(v)$ should be adjacent to vertices of each color class (except its color).
Otherwise, the colors of the vertices of $N(v)$ could be changed in order that some color $c$ no longer appear in $N(v)$, and consequently $v$ can be recolored with color $c$. Then, the graph induced by the vertices at distance at most $2$ from $v$ contains a b-$(t'+1)$-atom where $N[v]$ contains the center of this b-$(t'+1)$-atom.
Moreover, this b-$(t'+1)$-atom is feasible as the whole graph is $b$-$(t'+1)$-colorable, contradicting the hypothesis.
\end{proof}
\begin{prop}\label{bsensd1}
Let $G$ be a graph.
If $\varphi(G)= t$, then $G$ contains an induced feasible minimal b-$t$-atom and no induced feasible minimal b-$t'$-atom, for $t'>t$.
\end{prop}
\begin{proof}
Suppose $\varphi(G)=t$.
By Proposition \ref{bsensd}, $G$ contains an induced minimal b-$t$-atom. If no induced minimal b-$t$-atom is feasible, then there exists no b-$t$-coloring of $G$, a contradiction.
\end{proof}

A direct consequence of Proposition \ref{bsensi} and Proposition \ref{bsensd1} is the following.

\begin{theo}
For a graph $G$, we have $\varphi(G)=t$ if and only if $G$ contains an induced feasible minimal b-$t$-atom and no induced feasible minimal b-$t'$-atom, for $t'>t$.
\end{theo}

The following proposition will be useful in the last section.
\begin{prop}\label{poutineroidumonde}
Let $G$ be a graph and let $t=\varphi_r(G)$.
If every minimal b-$t$-atom is feasible in $G$, then $\varphi(G)=\varphi_r(G)$.
\end{prop}
\begin{proof}
Since $t=\varphi_r(G)$, $G$ does not contain a b-$(t+1)$-atom.
Thus, by Proposition \ref{bsensi}, we obtain $\varphi(G)=t$.
\end{proof}
Note that the problem of determining if a graph has a b-$t$-coloring is NP-complete even if $t$ is fixed \cite{PHDSA}.
However, it does not imply that determining if $\varphi(G)\ge t$ for a graph $G$ is NP-complete.
In contrast with the b-chromatic number, determining if a graph has b-relaxed number at least $t$ is in XP.
\begin{prop}
Let $G$ be a graph of order $n$ and let $t$ be an integer.
There exists an algorithm in time $O(n^{t^2})$ to determine if $\varphi_r(G)\ge t$. In particular, the problem b-r-COL with parameter $t$ is in XP.
\end{prop}
\begin{proof}
By Theorem \ref{bsensd2}, it suffices to verify that $G$ contains an induced minimal b-$t$-atom to determine if $\varphi_r (G)\ge t$.
By Observation \ref{taillbatom}, the order of a minimal b-$t$-atom is bounded by $t^2$. Thus, we obtain an algorithm in time $O(n^{t^2})$.
\end{proof}
Another NP-complete problem is to determine the \emph{b-spectrum} of a graph $G$ \cite{BA2007}, i.e. the set of integers $k$ such that $G$ is b-$k$-colorable.
For a graph $G$ satisfying $\varphi(G)=\varphi_r(G)$, our algorithm can be used. Thus, proving that for a class of graphs, every graph $G$ satisfies $\varphi(G)=\varphi_r(G)$, implies that the problem b-COL with parameter $t$ is in XP for this class of graphs.
\section{b-critical vertices and edges}
The concept of \emph{b-critical vertices} and \emph{b-critical edges} has been introduced recently and since five years a large number of articles are considering this subject \cite{BA2013,BL2013,BL2014,IK2010,ZAM2015}. In this section, we illustrate how this notion is strongly connected with the concept of b-$t$-atom.
\begin{de}[\cite{BL2013,IK2010}]
Let $G$ be a graph. A vertex $v$ of $G$ is b-critical if $\varphi(G-v)<\varphi(G)$. An edge $e$ is b-critical if $\varphi(G-e)<\varphi(G)$.
A vertex $v$ (edge $e$, respectively) in a graph $G$ is a \emph{b-$t$-trap}, if there exists a b-$t$-atom of $G$ that becomes feasible by removing $v$ ($e$, respectively). 
\end{de}
\begin{prop}\label{propcrit1}
Let $G$ be a graph. A vertex $v$ is b-critical if and only if it is in every feasible minimal b-$\varphi(G)$-atom and $v$ is not a b-$\varphi(G)$-trap.
\end{prop}
\begin{proof}
Let $t=\varphi(G)$. First, if $v$ is not in a feasible minimal b-$t$-atom, then $\varphi(G-v)=t$ and $v$ is not b-critical.
If $v$ is a b-$t$-trap, then, by definition, $\varphi(G-v)=t$.
Second, suppose $v$ is not a b-$t$-trap. 
If $v$ is in every feasible minimal b-$t$-atom, then, since every minimal b-$t$-atom in $G$ does not contain any other feasible minimal b-$t$-atom as induced subgraph, $G-v$ does not contain a feasible minimal b-$t$-atom. Thus, $v$ is b-critical.
\end{proof}
\begin{cor}
If a graph $G$ contains two induced feasible minimal b-$\varphi(G)$-atoms with disjoint set of vertices, then it contains no b-critical vertex.
\end{cor}

\begin{prop}\label{cogenial2}
Let $G$ be a graph and $v$ be a vertex of $V(G)$.
If $\varphi(G-v)> \varphi(G)$, then $G$ contains a minimal b-$\varphi(G-v)$-atom which is not feasible.
If $\varphi(G-v)< \varphi(G)-1$, then $G-v$ contains no feasible minimal b-$t$-atom, for $\varphi(G-v)< t\le \varphi(G)$.
\end{prop}
\begin{proof}
Note that every b-$t$-atom contained in $G-v$ is also contained in $G$, for any integer $t$. Thus, if $\varphi(G-v)> \varphi(G)$, then $G$ contains a b-$\varphi(G-v)$-trap and consequently a minimal b-$\varphi(G-v)$-atom which is not feasible.
Moreover, if $\varphi(G-v)< \varphi(G)-1$ and $G-v$ contains a feasible b-$t$-atom for $\varphi(G-v)< t\le \varphi(G)$, then $\varphi(G-v)\ge t$.
\end{proof}
In \cite{BA2013}, Balakrishnan and Raj have proved the following theorem.
\begin{theo}[\cite{BA2013}] \label{cogenial}
Let $G$ be a graph and $v$ be a vertex of $V(G)$. We have $\varphi(G)- \lfloor \frac{|V(G)|}{2}\rfloor+2\le \varphi(G-v)\le \varphi(G)+ \lfloor \frac{|V(G)|}{2}\rfloor-2$.
\end{theo}
Moreover, they have determined the families of graphs for which there exists a vertex $v$ such that $\varphi(G-v)=\varphi(G)- \lfloor \frac{|V(G)|}{2}\rfloor+2$ or $\varphi(G-v)= \varphi(G)+ \lfloor \frac{|V(G)|}{2}\rfloor-2$. In contrast with the b-chromatic number, we have the following property about the b-relaxed number.
\begin{prop}\label{vcritic}
Let $G$ be a graph. If a vertex $v$ is b-critical, then $\varphi_r(G-v)= \varphi_r (G)-1$.
\end{prop}
\begin{proof}
By Proposition \ref{propcrit1}, $v$ is in every b-$\varphi(G)$-atom. Let $i$ be the integer associated to $v$ in the construction of this b-$\varphi(G)$-atom. By removing the vertices with associated integer $i$, we obtain a b-$(\varphi(G)-1)$-atom and thus $\varphi_r(G-v)= \varphi_r(G)-1$.
\end{proof}
Note that this proposition was already proved for trees \cite{BL2013}.

\begin{lem}\label{lembc}
Let $G$ be a graph with $4\le |V(G)|\le 5$ and $E(G)\neq\emptyset$. We have $\varphi_r(G-v)=\varphi_r(G)+ \lfloor \frac{|V(G)|}{2}\rfloor-2$, for every vertex $v$ of $V(G)$, if and only if $G$ contains two disjoint edges but no induced minimal b-3-atom.
\end{lem}
\begin{proof}
We can note that we have $\varphi_r(G-v)=\varphi_r(G)+ \lfloor \frac{|V(G)|}{2}\rfloor-2$ if and only if $\varphi_r(G-v)=\varphi_r(G)$.

First, if $G$ contains no minimal b-3-atom and contains an edge, then $\varphi_r(G)=2$. Moreover, if $G$ contains two disjoint edges, then for any vertex $v$, $G-v$ contains $P_2$ and $\varphi_r(G-v)=2$.

Second, suppose that for every vertex $v$, $\varphi_r(G-v)=\varphi_r(G)$.
The only minimal b-3-atoms that contains at most five vertices are $K_3$, $C_5$ and $P_5$. Moreover, the only minimal b-4-atoms and b-5-atoms that contain at most five vertices are $K_4$ and $K_5$.
We are going to show that $G$ is not one of these graphs
\begin{description}
\item[Case 1:] $\varphi_r(G)=5$. If $G$ is a $K_5$, then, by removing any vertex $v$, we obtain $\varphi_r(G-v)=4$.
\item[Case 2:] $\varphi_r(G)=4$. If $G$ is a $K_4$, then, by removing any vertex $v$, we obtain $\varphi_r(G-v)=3$. If $G$ contains an induced $K_4$, $|V(G)|=5$ and $G$ is not $K_5$, then there exists a vertex $v$ such $G-v$ has no induced $K_4$ and $\varphi_r(G-v)=3$.
\item[Case 3:] $\varphi_r(G)=3$. If $G$ contains an induced $K_3$ and no induced $K_4$, then, since the induced $K_3$ in $G$ have a common vertex $v$, we obtain $\varphi_r(G-v)=2$. Moreover, if $G$ is $P_5$ or $C_5$, then, by removing any vertex $v$, we obtain $\varphi_r(G-v)=2$.
\end{description}
Thus, we can suppose that $\varphi_r(G)=2$.
If $G$ contains only edges with a common vertex $v$, then $\varphi_r(G-v)=1$. Hence, $G$ contains no b-3-atom and contain two disjoint edges.
\end{proof}
The following theorem is a generalization of a conjecture of Blidia et al. \cite{BL2012} for the parameter $\varphi_r$. Note that the graphs $P_4$, $C_4$ and $P_2 \cup P_2$ do not contain any induced minimal b-3-atom and contain two disjoint edges.
\begin{theo}
Let $G$ be a graph. We have $\varphi_r(G-v)=\varphi_r(G)+ \lfloor \frac{|V(G)|}{2}\rfloor-2$, for every vertex $v$ of $V(G)$, if and only one of these conditions is true about $G$:
\begin{enumerate}
\item[i)] $G$ is $P_2$ or $C_3$.
\item[ii)] $E(G)=\emptyset$ and $4\le |V(G)|\le 5$.
\item[iii)] $4\le |V(G)|\le 5$ and $G$ contains two disjoint edges but no b-3-atom.
\end{enumerate}
\end{theo}
\begin{proof}
Note that if $|V(G)|\ge 6$, then, by Proposition \ref{vcritic}, we can not have $\varphi_r(G-v)= \varphi_r(G)+ \lfloor |V(G)|/ 2\rfloor-2$.
Note also that if $G$ contains only one vertex, then it can not satisfy $\varphi_r(G-v)= \varphi_r(G)+ \lfloor |V(G)|/ 2\rfloor-2$.

First, if $2 \le |V(G)|\le 3$, then we have $\varphi_r(G-v)=\varphi_r(G)-1$ if and only if $G$ is a minimal b-$t$-atom. Hence, if and only if $G$ is $P_2$ or $C_3$.
Second, if $G$ contains no edges, then $\varphi_r(G)=1$ and for any vertex $v$, $\varphi_r(G-v)=1$.
The third condition is obtained by Lemma \ref{lembc}.
\end{proof}
\begin{de}
Let $t$ be a positive integer and $A$ be a b-$t$-atom. An edge $e$ is \emph{b-atom-critical} in $A$ if $A-e$ is not a b-$t$-atom.
\end{de}
\begin{prop}\label{propcrit2}
Let $G$ be a graph. An edge $e$ is b-critical if and only if it is b-atom-critical in every feasible minimal b-$\varphi(G)$-atom and $e$ is not a b-$\varphi(G)$-trap.
\end{prop}
\begin{proof}
Let $t=\varphi(G)$.
First, if $e$ is not b-atom-critical in a feasible minimal b-$t$-atom, then $G-e$ contains a feasible minimal b-$t$-atom and $\varphi(G-e)=t$.
If $e$ is a b-$t$-trap, then, by definition, $\varphi(G-e)=t$.
Second, suppose that $e$ is not a b-$t$-trap.
If $e$ is b-atom-critical in every feasible minimal b-$t$-atom, then, since every feasible minimal b-$t$-atom in $G$ does not contain any other feasible minimal b-$t$-atom as subgraph in $G-e$, the graph $G-e$ does not contain a feasible minimal b-$t$-atom. Thus, $e$ is b-critical.
\end{proof}

\begin{cor}
If a graph $G$ contains two induced feasible minimal b-$\varphi(G)$-atoms with disjoint sets of b-atom-critical edges, then $G$ contains no b-critical edge.
\end{cor}
\section{b-perfect graphs}
A b-perfect graph is a graph for which every induced subgraph satisfies that its b-chromatic number is equal to its chromatic number. More generally, we present the following definitions.
\begin{de}[\cite{HO2005}]
A graph $G$ is \emph{b-$\chi$-$k$-bounded}, for $k$ a positive integer, if $ \varphi(G')-\chi(G') \le k$, for every induced subgraph $G'$ of $G$.
A graph $G$ is a \emph{$\chi$-$k$-unbounded b-atom}, for $k$ a positive integer, if $ \varphi(G)-\chi(G) > k$ and $G$ is a b-$t$-atom for some integer $t$.
A graph $G$ is an \emph{imperfect b-atom}, for $k$ a positive integer, if $ \varphi(G)>\chi(G)$ and $G$ is a b-$t$-atom for some integer $t$.
\end{de}
\begin{figure}[t]
\begin{center}
\begin{tikzpicture}
\draw (0,0) -- (0,10);
\draw (2,0) -- (2,10);
\draw (4,0) -- (4,10);
\draw (6,0) -- (6,10);
\draw (8,0) -- (8,10);
\draw (10,0) -- (10,10);
\draw (0,0) -- (10,0);
\draw (0,2) -- (10,2);
\draw (0,4) -- (10,4);
\draw (0,6) -- (10,6);
\draw (0,8) -- (10,8);
\draw (0,10) -- (10,10);
\node at (1,8.2) {$F_{1}$};
\draw (0.4,8.8) -- (0.7,9.3);
\draw (0.7,9.3) -- (1,8.8);
\draw (1,8.8) -- (1.3,9.3);
\draw (1.3,9.3) -- (1.6,8.8);
\node at (0.4,8.8) [circle,draw=black,fill=black, scale=0.3] {};
\node at (1,8.8) [circle,draw=black,fill=black, scale=0.3] {};
\node at (1.6,8.8) [circle,draw=black,fill=black, scale=0.3] {};
\node at (0.7,9.3) [circle,draw=black,fill=black, scale=0.3] {};
\node at (1.3,9.3) [circle,draw=black,fill=black, scale=0.3] {};
\node at (3,8.2) {$F_{2}$};
\draw (2.2,8.8) -- (2.5,9.3);
\draw (2.5,9.3) -- (2.8,8.8);
\draw (3,8.8) -- (3,9.3);
\draw (3.5,9.3) -- (3,9.3);
\draw (3.5,9.3) -- (3.5,8.8);
\node at (2.2,8.8) [circle,draw=black,fill=black, scale=0.3] {};
\node at (2.5,9.3) [circle,draw=black,fill=black, scale=0.3] {};
\node at (2.8,8.8) [circle,draw=black,fill=black, scale=0.3] {};
\node at (3,8.8) [circle,draw=black,fill=black, scale=0.3] {};
\node at (3,9.3) [circle,draw=black,fill=black, scale=0.3] {};
\node at (3.5,8.8) [circle,draw=black,fill=black, scale=0.3] {};
\node at (3.5,9.3) [circle,draw=black,fill=black, scale=0.3] {};
\node at (5,8.2) {$F_{3}$};
\draw (4.5,8.5) -- (4.5,9.5);
\draw (5,8.5) -- (5,9.5);
\draw (5.5,8.5) -- (5.5,9.5);
\node at (4.5,8.5) [circle,draw=black,fill=black, scale=0.3] {};
\node at (4.5,9) [circle,draw=black,fill=black, scale=0.3] {};
\node at (4.5,9.5) [circle,draw=black,fill=black, scale=0.3] {};
\node at (5,8.5) [circle,draw=black,fill=black, scale=0.3] {};
\node at (5,9) [circle,draw=black,fill=black, scale=0.3] {};
\node at (5,9.5) [circle,draw=black,fill=black, scale=0.3] {};
\node at (5.5,8.5) [circle,draw=black,fill=black, scale=0.3] {};
\node at (5.5,9) [circle,draw=black,fill=black, scale=0.3] {};
\node at (5.5,9.5) [circle,draw=black,fill=black, scale=0.3] {};
\node at (7,8.2) {$F_{4}$};
\draw (6.75,8.75) -- (6.75,9.25);
\draw (7.25,8.75) -- (7.25,9.25);
\draw (6.75,8.75) -- (7.25,8.75);
\draw (6.75,9.25) -- (7.25,9.25);
\draw (6.75,8.75) -- (7.25,9.25);
\draw (6.75,8.75) -- (6.25,9);
\draw (6.75,9.25) -- (6.25,9);
\draw (7.25,8.75) -- (7.75,9);
\draw (7.25,9.25) -- (7.75,9);
\node at (6.75,8.75) [circle,draw=black,fill=black, scale=0.3] {};
\node at (6.75,9.25) [circle,draw=black,fill=black, scale=0.3] {};
\node at (7.25,8.75) [circle,draw=black,fill=black, scale=0.3] {};
\node at (7.25,9.25) [circle,draw=black,fill=black, scale=0.3] {};
\node at (6.25,9) [circle,draw=black,fill=black, scale=0.3] {};
\node at (7.75,9) [circle,draw=black,fill=black, scale=0.3] {};
\node at (9,8.2) {$F_{5}$};
\draw (9,9) -- (9.4,9.4);
\draw (9,9) -- (8.6,9.4);
\draw (9,9) -- (9.4,8.6);
\draw (9,9) -- (8.6,8.6);
\draw (9.8,9) -- (9.4,9.4);
\draw (8.2,9) -- (8.6,9.4);
\draw (9.8,9) -- (9.4,8.6);
\draw (8.2,9) -- (8.6,8.6);
\draw (9,9) -- (9.8,9);
\draw (8.6,8.6) -- (8.6,9.4);
\node at (9,9) [circle,draw=black,fill=black, scale=0.3] {};
\node at (9.4,9.4) [circle,draw=black,fill=black, scale=0.3] {};
\node at (8.6,9.4) [circle,draw=black,fill=black, scale=0.3] {};
\node at (9.4,8.6) [circle,draw=black,fill=black, scale=0.3] {};
\node at (8.6,8.6) [circle,draw=black,fill=black, scale=0.3] {};
\node at (9.8,9) [circle,draw=black,fill=black, scale=0.3] {};
\node at (8.2,9) [circle,draw=black,fill=black, scale=0.3] {};
\node at (1,6.2) {$F_{6}$};
\draw (0.2,7) -- (0.9,7);
\draw (0.2,7) -- (0.55,7.35);
\draw (0.9,7) -- (0.55,7.35);
\draw (0.2,7) -- (0.55,6.65);
\draw (0.9,7) -- (0.55,6.65);
\draw (1.1,7) -- (1.8,7);
\draw (1.1,7) -- (1.45,7.35);
\draw (1.8,7) -- (1.45,7.35);
\draw (1.1,7) -- (1.45,6.65);
\draw (1.8,7) -- (1.45,6.65);
\node at (0.2,7) [circle,draw=black,fill=black, scale=0.3] {};
\node at (0.9,7) [circle,draw=black,fill=black, scale=0.3] {};
\node at (0.55,7.35) [circle,draw=black,fill=black, scale=0.3] {};
\node at (0.55,6.65) [circle,draw=black,fill=black, scale=0.3] {};
\node at (1.1,7) [circle,draw=black,fill=black, scale=0.3] {};
\node at (1.8,7) [circle,draw=black,fill=black, scale=0.3] {};
\node at (1.45,7.35) [circle,draw=black,fill=black, scale=0.3] {};
\node at (1.45,6.65) [circle,draw=black,fill=black, scale=0.3] {};
\node at (3,6.2) {$F_{7}$};
\draw (2.2,7) -- (2.9,7);
\draw (2.2,7) -- (2.55,7.35);
\draw (2.9,7) -- (2.55,7.35);
\draw (2.2,7) -- (2.55,6.65);
\draw (2.9,7) -- (2.55,6.65);
\draw (2.9,7) -- (3.8,7);
\draw (3.1,7) -- (3.45,7.35);
\draw (3.8,7) -- (3.45,7.35);
\draw (3.1,7) -- (3.45,6.65);
\draw (3.8,7) -- (3.45,6.65);
\node at (2.2,7) [circle,draw=black,fill=black, scale=0.3] {};
\node at (2.9,7) [circle,draw=black,fill=black, scale=0.3] {};
\node at (2.55,7.35) [circle,draw=black,fill=black, scale=0.3] {};
\node at (2.55,6.65) [circle,draw=black,fill=black, scale=0.3] {};
\node at (3.1,7) [circle,draw=black,fill=black, scale=0.3] {};
\node at (3.8,7) [circle,draw=black,fill=black, scale=0.3] {};
\node at (3.45,7.35) [circle,draw=black,fill=black, scale=0.3] {};
\node at (3.45,6.65) [circle,draw=black,fill=black, scale=0.3] {};
\node at (5,6.2) {$F_{8}$};
\draw (4.4,6.5) -- (5.6,6.5);
\draw (4.4,6.5) -- (5,6.9);
\draw (4.8,6.5) -- (5,6.9);
\draw (5.2,6.5) -- (5,6.9);
\draw (5.6,6.5) -- (5,6.9);
\draw (5,6.9) -- (5,7.3);
\draw (5,7.3) -- (4.6,7.7);
\draw (5,7.3) -- (5.4,7.7);
\draw (4.6,7.7) -- (5.4,7.7);
\node at (4.4,6.5) [circle,draw=black,fill=black, scale=0.3] {};
\node at (4.8,6.5) [circle,draw=black,fill=black, scale=0.3] {};
\node at (5.2,6.5) [circle,draw=black,fill=black, scale=0.3] {};
\node at (5.6,6.5) [circle,draw=black,fill=black, scale=0.3] {};
\node at (5,6.9) [circle,draw=black,fill=black, scale=0.3] {};
\node at (5,7.3) [circle,draw=black,fill=black, scale=0.3] {};
\node at (5.4,7.7) [circle,draw=black,fill=black, scale=0.3] {};
\node at (4.6,7.7) [circle,draw=black,fill=black, scale=0.3] {};
\node at (7,6.2) {$F_{9}$};
\draw (6.2,6.5) -- (6.9,6.5);
\draw (7.1,6.5) -- (7.8,6.5);
\draw (6.2,6.5) -- (7,6.9);
\draw (6.55,6.5) -- (7,6.9);
\draw (6.9,6.5) -- (7,6.9);
\draw (7.1,6.5) -- (7,6.9);
\draw (7.45,6.5) -- (7,6.9);
\draw (7.8,6.5) -- (7,6.9);
\draw (7,6.9) -- (7,7.3);
\draw (7,7.3) -- (6.6,7.7);
\draw (7,7.3) -- (7.4,7.7);
\draw (6.6,7.7) -- (7.4,7.7);
\node at (6.2,6.5) [circle,draw=black,fill=black, scale=0.3] {};
\node at (6.55,6.5) [circle,draw=black,fill=black, scale=0.3] {};
\node at (6.9,6.5) [circle,draw=black,fill=black, scale=0.3] {};
\node at (7.1,6.5) [circle,draw=black,fill=black, scale=0.3] {};
\node at (7.45,6.5) [circle,draw=black,fill=black, scale=0.3] {};
\node at (7.8,6.5) [circle,draw=black,fill=black, scale=0.3] {};
\node at (7,6.9) [circle,draw=black,fill=black, scale=0.3] {};
\node at (7,7.3) [circle,draw=black,fill=black, scale=0.3] {};
\node at (7.4,7.7) [circle,draw=black,fill=black, scale=0.3] {};
\node at (6.6,7.7) [circle,draw=black,fill=black, scale=0.3] {};
\node at (9,6.2) {$F_{10}$};
\draw (8.7,6.9) -- (9.3,6.9);
\draw (8.7,7.5) -- (9.3,7.5);
\draw (8.7,6.9) -- (8.7,7.5);
\draw (9.3,6.9) -- (9.3,7.5);
\draw (8.7,6.9) -- (9.3,7.5);
\draw (9.3,6.9) -- (9.7,6.5);
\draw (9.3,7.5) -- (9.7,6.5);
\draw (8.7,6.9) -- (8.3,6.5);
\draw (8.7,7.5) -- (8.3,6.5);
\draw (9.7,6.5) -- (8.3,6.5);
\node at (9.7,6.5) [circle,draw=black,fill=black, scale=0.3] {};
\node at (8.3,6.5) [circle,draw=black,fill=black, scale=0.3] {};
\node at (8.7,6.9) [circle,draw=black,fill=black, scale=0.3] {};
\node at (8.7,7.5) [circle,draw=black,fill=black, scale=0.3] {};
\node at (9.3,6.9) [circle,draw=black,fill=black, scale=0.3] {};
\node at (9.3,7.5) [circle,draw=black,fill=black, scale=0.3] {};
\node at (1,4.2) {$F_{11}$};
\draw (1,5) -- (1.4,5.4);
\draw (1,5) -- (0.6,5.4);
\draw (1,5) -- (1.4,4.6);
\draw (1,5) -- (0.6,4.6);
\draw (1.8,5) -- (1.4,5.4);
\draw (0.2,5) -- (0.6,5.4);
\draw (1.8,5) -- (1.4,4.6);
\draw (0.2,5) -- (0.6,4.6);
\draw (1,5) -- (1.8,5);
\draw (0.6,4.6) -- (0.6,5.4);
\draw (0.2,5) -- (1.4,5.4);
\draw (0.2,5) -- (1.4,4.6);
\node at (1,5) [circle,draw=black,fill=black, scale=0.3] {};
\node at (1.4,5.4) [circle,draw=black,fill=black, scale=0.3] {};
\node at (0.6,5.4) [circle,draw=black,fill=black, scale=0.3] {};
\node at (1.4,4.6) [circle,draw=black,fill=black, scale=0.3] {};
\node at (0.6,4.6) [circle,draw=black,fill=black, scale=0.3] {};
\node at (1.8,5) [circle,draw=black,fill=black, scale=0.3] {};
\node at (0.2,5) [circle,draw=black,fill=black, scale=0.3] {};
\node at (3,4.2) {$F_{12}$};
\draw (2.2,4.7) -- (3.8,4.7);
\draw (2.2,5.3) -- (3.8,5.3);
\draw (2.2,4.7) -- (2.8,5.3);
\draw (2.8,4.7) -- (2.2,5.3);
\draw (3.2,4.7) -- (3.8,5.3);
\draw (3.8,4.7) -- (3.2,5.3);
\draw (2.8,4.7) -- (2.8,5.3);
\draw (3.2,4.7) -- (3.2,5.3);
\node at (2.2,5.3) [circle,draw=black,fill=black, scale=0.3] {};
\node at (2.8,5.3) [circle,draw=black,fill=black, scale=0.3] {};
\node at (3.2,5.3) [circle,draw=black,fill=black, scale=0.3] {};
\node at (3.8,5.3) [circle,draw=black,fill=black, scale=0.3] {};
\node at (2.2,4.7) [circle,draw=black,fill=black, scale=0.3] {};
\node at (2.8,4.7) [circle,draw=black,fill=black, scale=0.3] {};
\node at (3.2,4.7) [circle,draw=black,fill=black, scale=0.3] {};
\node at (3.8,4.7) [circle,draw=black,fill=black, scale=0.3] {};
\node at (5,4.2) {$F_{13}$};
\draw (4.1,4.7) -- (5.9,4.7);
\draw (4.1,5.3) -- (5.9,5.3);
\draw (4.1,4.7) -- (4.7,5.3);
\draw (4.7,4.7) -- (4.1,5.3);
\draw (5.3,4.7) -- (5.9,5.3);
\draw (5.9,4.7) -- (5.3,5.3);
\draw (4.7,4.7) -- (5.3,5.3);
\draw (5.3,4.7) -- (4.7,5.3);
\draw (4.1,4.7) -- (4.1,5.3);
\draw (5.9,4.7) -- (5.9,5.3);
\node at (4.1,5.3) [circle,draw=black,fill=black, scale=0.3] {};
\node at (4.7,5.3) [circle,draw=black,fill=black, scale=0.3] {};
\node at (5.3,5.3) [circle,draw=black,fill=black, scale=0.3] {};
\node at (5.9,5.3) [circle,draw=black,fill=black, scale=0.3] {};
\node at (4.1,4.7) [circle,draw=black,fill=black, scale=0.3] {};
\node at (4.7,4.7) [circle,draw=black,fill=black, scale=0.3] {};
\node at (5.3,4.7) [circle,draw=black,fill=black, scale=0.3] {};
\node at (5.9,4.7) [circle,draw=black,fill=black, scale=0.3] {};
\node at (7,4.2) {$F_{14}$};
\draw (6.1,4.6) -- (7.9,4.6);
\draw (6.1,5.4) -- (7.9,5.4);
\draw (6.7,5.2) -- (7.3,5.2);
\draw (6.7,5.2) -- (6.1,5.4);
\draw (7.3,5.2) -- (7.9,5.4);
\draw (6.1,4.6) -- (6.7,5.2);
\draw (6.7,4.6) -- (6.1,5.4);
\draw (7.3,4.6) -- (7.9,5.4);
\draw (7.9,4.6) -- (7.3,5.2);
\draw (6.7,4.6) -- (7.3,5.2);
\draw (7.3,4.6) -- (6.7,5.2);
\draw (6.1,4.6) -- (6.1,5.4);
\draw (7.9,4.6) -- (7.9,5.4);
\node at (6.1,5.4) [circle,draw=black,fill=black, scale=0.3] {};
\node at (6.7,5.2) [circle,draw=black,fill=black, scale=0.3] {};
\node at (7.3,5.2) [circle,draw=black,fill=black, scale=0.3] {};
\node at (7.9,5.4) [circle,draw=black,fill=black, scale=0.3] {};
\node at (6.1,4.6) [circle,draw=black,fill=black, scale=0.3] {};
\node at (6.7,4.6) [circle,draw=black,fill=black, scale=0.3] {};
\node at (7.3,4.6) [circle,draw=black,fill=black, scale=0.3] {};
\node at (7.9,4.6) [circle,draw=black,fill=black, scale=0.3] {};
\node at (9,4.2) {$F_{15}$};
\draw (8.1,4.5) -- (9.9,4.5);
\draw (8.1,4.5) -- (8.7,4.7);
\draw (9.3,4.7) -- (9.9,4.5);
\draw (8.7,4.7) -- (9.3,4.7);
\draw (8.1,5.5) -- (9.9,5.5);
\draw (8.7,5.3) -- (9.3,5.3);
\draw (8.7,5.3) -- (8.1,5.5);
\draw (9.3,5.3) -- (9.9,5.5);
\draw (8.1,4.5) -- (8.7,5.3);
\draw (8.7,4.7) -- (8.1,5.5);
\draw (9.3,4.7) -- (9.9,5.5);
\draw (9.9,4.5) -- (9.3,5.3);
\draw (8.7,4.7) -- (9.3,5.3);
\draw (9.3,4.7) -- (8.7,5.3);
\draw (8.1,4.5) -- (8.1,5.5);
\draw (9.9,4.5) -- (9.9,5.5);
\node at (8.1,5.5) [circle,draw=black,fill=black, scale=0.3] {};
\node at (8.7,5.3) [circle,draw=black,fill=black, scale=0.3] {};
\node at (9.3,5.3) [circle,draw=black,fill=black, scale=0.3] {};
\node at (9.9,5.5) [circle,draw=black,fill=black, scale=0.3] {};
\node at (8.1,4.5) [circle,draw=black,fill=black, scale=0.3] {};
\node at (8.7,4.7) [circle,draw=black,fill=black, scale=0.3] {};
\node at (9.3,4.7) [circle,draw=black,fill=black, scale=0.3] {};
\node at (9.9,4.5) [circle,draw=black,fill=black, scale=0.3] {};
\node at (1,2.2) {$F_{16}$};
\draw (0.7,2.5) -- (1.3,2.5);
\draw (0.7,2.5) -- (0.15,3.3);
\draw (1.3,2.5) -- (1.85,3.3);
\draw (1,3.9) -- (0.15,3.3);
\draw (1,3.9) -- (1.85,3.3);
\draw (1,3.1) -- (0.15,3.3);
\draw (1,3.1) -- (1.85,3.3);
\draw (1,3.1) -- (1,3.9);
\node at (0.7,2.5) [circle,draw=black,fill=black, scale=0.3] {};
\node at (1.3,2.5) [circle,draw=black,fill=black, scale=0.3] {};
\node at (1.85,3.3) [circle,draw=black,fill=black, scale=0.3] {};
\node at (0.15,3.3) [circle,draw=black,fill=black, scale=0.3] {};
\node at (1,3.1) [circle,draw=black,fill=black, scale=0.3] {};
\node at (1,3.9) [circle,draw=black,fill=black, scale=0.3] {};
\node at (3,2.2) {$F_{17}$};
\draw (2.7,2.5) -- (3.3,2.5);
\draw (2.7,2.5) -- (2.15,3.3);
\draw (3.3,2.5) -- (3.85,3.3);
\draw (3,3.9) -- (2.15,3.3);
\draw (3,3.9) -- (3.85,3.3);
\draw (3,3.1) -- (2.15,3.3);
\draw (3,3.1) -- (3.85,3.3);
\draw (3,3.1) -- (3,3.9);
\draw (3,3.1) -- (3.3,2.5);
\node at (2.7,2.5) [circle,draw=black,fill=black, scale=0.3] {};
\node at (3.3,2.5) [circle,draw=black,fill=black, scale=0.3] {};
\node at (3.85,3.3) [circle,draw=black,fill=black, scale=0.3] {};
\node at (2.15,3.3) [circle,draw=black,fill=black, scale=0.3] {};
\node at (3,3.1) [circle,draw=black,fill=black, scale=0.3] {};
\node at (3,3.9) [circle,draw=black,fill=black, scale=0.3] {};
\node at (5,2.2) {$F_{18}$};
\draw (4.7,2.5) -- (5.3,2.5);
\draw (4.7,2.5) -- (4.15,3.3);
\draw (5.3,2.5) -- (5.85,3.3);
\draw (5,3.9) -- (4.15,3.3);
\draw (5,3.9) -- (5.85,3.3);
\draw (4.7,2.5) -- (4.7,3);
\draw (5.3,2.5) -- (5.3,3);
\draw (4.7,3) -- (5,3.4);
\draw (5.3,3) -- (5,3.4);
\draw (4.15,3.3) -- (4.7,3);
\draw (4.15,3.3) -- (5.3,3);
\draw (5.85,3.3) -- (4.7,3);
\draw (5.85,3.3) -- (5.3,3);
\node at (4.7,2.5) [circle,draw=black,fill=black, scale=0.3] {};
\node at (5.3,2.5) [circle,draw=black,fill=black, scale=0.3] {};
\node at (5.85,3.3) [circle,draw=black,fill=black, scale=0.3] {};
\node at (4.15,3.3) [circle,draw=black,fill=black, scale=0.3] {};
\node at (4.7,3) [circle,draw=black,fill=black, scale=0.3] {};
\node at (5.3,3) [circle,draw=black,fill=black, scale=0.3] {};
\node at (5,3.4) [circle,draw=black,fill=black, scale=0.3] {};
\node at (5,3.9) [circle,draw=black,fill=black, scale=0.3] {};
\node at (7,2.2) {$F_{19}$};
\draw (6.7,2.5) -- (7.3,2.5);
\draw (6.7,2.5) -- (6.15,3.3);
\draw (7.3,2.5) -- (7.85,3.3);
\draw (7,3.9) -- (6.15,3.3);
\draw (7,3.9) -- (7.85,3.3);
\draw (6.7,2.5) -- (6.7,3);
\draw (7.3,2.5) -- (7.3,3);
\draw (6.7,3) -- (7,3.4);
\draw (7.3,3) -- (7,3.4);
\draw (7.3,3) -- (6.7,3);
\draw (6.7,3) -- (7,3.9);
\draw (7.3,3) -- (7,3.9);
\draw (7,3.4) -- (6.15,3.3);
\draw (7,3.4) -- (7.85,3.3);
\node at (6.7,2.5) [circle,draw=black,fill=black, scale=0.3] {};
\node at (7.3,2.5) [circle,draw=black,fill=black, scale=0.3] {};
\node at (7.85,3.3) [circle,draw=black,fill=black, scale=0.3] {};
\node at (6.15,3.3) [circle,draw=black,fill=black, scale=0.3] {};
\node at (6.7,3) [circle,draw=black,fill=black, scale=0.3] {};
\node at (7.3,3) [circle,draw=black,fill=black, scale=0.3] {};
\node at (7,3.4) [circle,draw=black,fill=black, scale=0.3] {};
\node at (7,3.9) [circle,draw=black,fill=black, scale=0.3] {};
\node at (9,2.2) {$F_{20}$};
\draw (8.7,2.5) -- (9.3,2.5);
\draw (8.7,2.5) -- (8.15,3.3);
\draw (9.3,2.5) -- (9.85,3.3);
\draw (9,3.9) -- (8.15,3.3);
\draw (9,3.9) -- (9.85,3.3);
\draw (8.7,2.5) -- (9.3,3);
\draw (9.3,2.5) -- (8.7,3);
\draw (8.7,3) -- (9,3.4);
\draw (9.3,3) -- (8.7,3);
\draw (9.3,3) -- (9,3.9);
\draw (9,3.4) -- (8.15,3.3);
\draw (9,3.4) -- (9.85,3.3);
\draw (8.7,3) -- (8.15,3.3);
\draw (9.3,3) -- (9.85,3.3);
\node at (8.7,2.5) [circle,draw=black,fill=black, scale=0.3] {};
\node at (9.3,2.5) [circle,draw=black,fill=black, scale=0.3] {};
\node at (9.85,3.3) [circle,draw=black,fill=black, scale=0.3] {};
\node at (8.15,3.3) [circle,draw=black,fill=black, scale=0.3] {};
\node at (8.7,3) [circle,draw=black,fill=black, scale=0.3] {};
\node at (9.3,3) [circle,draw=black,fill=black, scale=0.3] {};
\node at (9,3.4) [circle,draw=black,fill=black, scale=0.3] {};
\node at (9,3.9) [circle,draw=black,fill=black, scale=0.3] {};
\node at (1,0.2) {$F_{21}$};
\draw (0.7,0.5) -- (1.3,0.5);
\draw (0.7,0.5) -- (0.15,1.3);
\draw (1.3,0.5) -- (1.85,1.3);
\draw (1,1.9) -- (0.15,1.3);
\draw (1,1.9) -- (1.85,1.3);
\draw (0.7,0.5) -- (1.3,1);
\draw (1.3,0.5) -- (0.7,1);
\draw (0.7,1) -- (1,1.4);
\draw (1.3,1) -- (0.7,1);
\draw (1.3,1) -- (1,1.9);
\draw (1,1.4) -- (0.15,1.3);
\draw (1,1.4) -- (1.85,1.3);
\draw (0.7,1) -- (0.15,1.3);
\draw (1.3,1) -- (1.85,1.3);
\draw (1,1.4) -- (0.7,0.5);
\node at (0.7,0.5) [circle,draw=black,fill=black, scale=0.3] {};
\node at (1.3,0.5) [circle,draw=black,fill=black, scale=0.3] {};
\node at (1.85,1.3) [circle,draw=black,fill=black, scale=0.3] {};
\node at (0.15,1.3) [circle,draw=black,fill=black, scale=0.3] {};
\node at (0.7,1) [circle,draw=black,fill=black, scale=0.3] {};
\node at (1.3,1) [circle,draw=black,fill=black, scale=0.3] {};
\node at (1,1.4) [circle,draw=black,fill=black, scale=0.3] {};
\node at (1,1.9) [circle,draw=black,fill=black, scale=0.3] {};
\node at (3,0.2) {$F_{22}$};
\draw (2.7,0.5) -- (3.3,0.5);
\draw (2.7,0.5) -- (2.15,1.3);
\draw (3.3,0.5) -- (3.85,1.3);
\draw (3,1.9) -- (2.15,1.3);
\draw (3,1.9) -- (3.85,1.3);
\draw (2.7,1.4) -- (2.15,1.3);
\draw (3.3,1.4) -- (3.85,1.3);
\draw (2.7,1.4) -- (3.3,1.4);
\draw (3.3,1.4) -- (3,1.9);
\draw (2.7,1.4) -- (3,1.9);
\draw (3.3,1.4) -- (2.7,0.5);
\draw (2.7,1.4) -- (3.3,0.5);
\draw (3.3,0.5) -- (3,0.7);
\draw (2.7,0.5) -- (3,0.7);

\draw (3,1.9) -- (3,0.7);
\node at (2.7,0.5) [circle,draw=black,fill=black, scale=0.3] {};
\node at (3.3,0.5) [circle,draw=black,fill=black, scale=0.3] {};
\node at (3.85,1.3) [circle,draw=black,fill=black, scale=0.3] {};
\node at (2.15,1.3) [circle,draw=black,fill=black, scale=0.3] {};
\node at (3,1.9) [circle,draw=black,fill=black, scale=0.3] {};
\node at (3.3,1.4) [circle,draw=black,fill=black, scale=0.3] {};
\node at (2.7,1.4) [circle,draw=black,fill=black, scale=0.3] {};
\node at (3,0.7) [circle,draw=black,fill=black, scale=0.3] {};
\end{tikzpicture}
\end{center}
\caption{The family $\mathcal{F}$: the imperfect b-atoms \cite{HO2012}.}
\label{bfamily}
\end{figure}
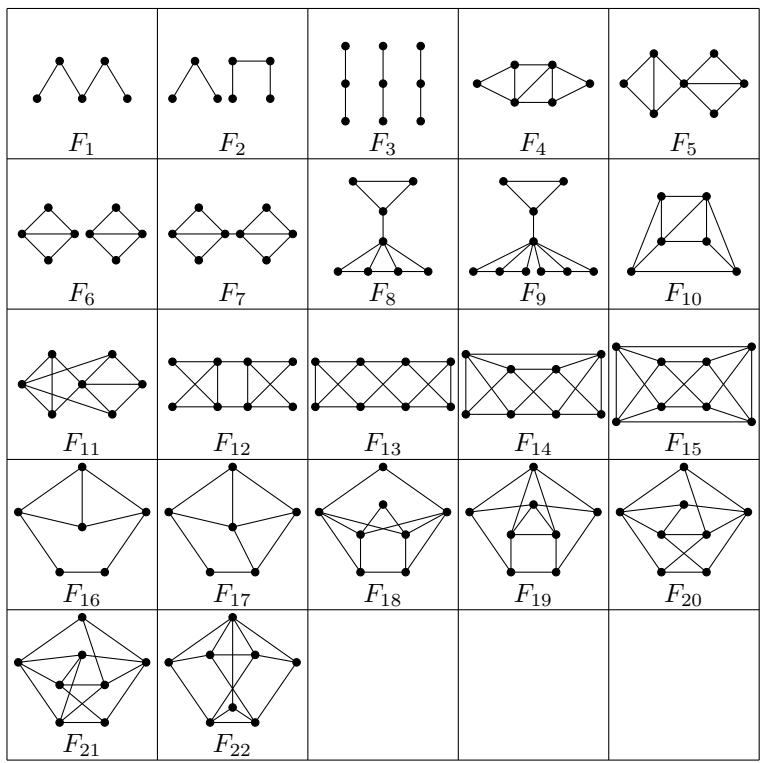

Hoang et al. \cite{HO2012} characterized b-perfect graphs by giving the family $\mathcal{F}$ of forbidden induced subgraphs depicted in Figure \ref{bfamily}.
We recall the following theorem:
\begin{theo}[\cite{HO2012}]\label{theoperfect}
A graph is b-perfect if and only if it contains no graph from $\mathcal{F}$ as induced subgraph.
\end{theo}
Note that every graph in the family $\mathcal{F}$ is a b-$t$-atom for some $t$. More precisely, $F_1$, $F_2$ and $F_3$ are the only minimal bipartite b-3-atoms. The remaining graphs are minimal b-4-atoms that do not contain $F_1$, $F_2$ and $F_3$ as induced subgraph and which admit a proper coloring with three colors (as mentioned in \cite{HO2009}). We can state the following property about b-$t$-atoms.
\begin{theo}\label{minnn}
Let $k$ be a positive integer.
A graph $G$ is not b-$\chi$-$k$-bounded if and only if it contains a minimal $\chi$-$k$-unbounded b-atom.
\end{theo}
\begin{proof}
First, if $G$ contains a minimal $\chi$-$k$-unbounded b-atom, then, by definition, $G$ is not $\chi$-$k$-bounded.

Second, suppose $G$ is not b-$\chi$-$k$-bounded. Then, there exists an induced subgraph $A$ of $G$ of minimal order which is not b-$\chi$-$k$-bounded. By removing vertices of $A$ we can only decrease the chromatic number.
Thus, by removing vertices we can obtain a b-$\varphi(A)$-atom which is $\chi$-$k$-unbounded.
\end{proof}
\begin{cor}\label{corrr}
The graphs with b-chromatic number $t$ which are b-$\chi$-$k$-bounded,, for fixed integers $k$ and $t$, can be defined by forbidding a finite family of induced subgraphs: the $\chi$-$k$-unbounded b-atoms.
Hence, a graph $G$ is b-perfect if and only if it does not contain imperfect b-atoms.
\end{cor}

Let b-$\chi$-BOUNDED be the following decision problem and let $k$ be an integer, with $0\le k< \varphi(G)$.
\begin{center}
\parbox{10cm}{
\setlength{\parskip}{.05cm}
\textbf{b-$\chi$-$k$-BOUNDED}

\textbf{Instance} : A graph $G$.

\textbf{Question}: Does $\varphi(G)-\chi(G)\ge k$?}

\end{center}
By Corollary \ref{corrr}, we obtain the following corollary:
\begin{cor}
Let $G$ be a graph and $k$ be an integer, with $0\le k< \varphi(G)$. 
There exists an algorithm in time $O(n^{\varphi(G)^2})$ to solve b-$\chi$-$k$-BOUNDED.
\end{cor}

Since a graph $G$ is b-perfect if and only if it does not contain imperfect b-atoms, we have the following theorem:
\begin{theo}
The number of imperfect b-atoms is finite. A graph is an imperfect b-atom if and only if it is in the family $\mathcal{F}$ ( Figure \ref{bfamily}).
\end{theo}
The previous theorem is a consequence of Theorem \ref{theoperfect}.
Remark that if we can prove that every minimal b-4-atom except $K_4$ contains an induced subgraph of the family $\mathcal{F}$, then, using Theorem \ref{minnn}, we obtain another proof of Theorem \ref{theoperfect}. 
\section{b-chromatic and b-relaxed chromatic numbers}
In this section we consider the b-relaxed number relatively to the b-chromatic number and prove equality for trees and graphs of girth ar least $7$.
\begin{lem}\label{prout}
A minimal b-$t$-atom has at most $t$ connected components.
\end{lem}
\begin{proof}
Suppose that a minimal b-$t$-atom $G$ has more than $t$ connected component. By definition, at least one connected component $A$ of $G$ does not contain a vertex of $C(G)$. Since $G-A$ is also a b-$t$-atom, $G$ is not minimal.
\end{proof}
Note that a minimal b-$t$-atom $G$ contains a center $C(G)$ and the remaining vertices of $G$ are neighbors of vertices of $C(G)$.
\begin{prop}\label{tree44}
For a tree $T$, we have $\varphi(T)=\varphi_r(T)$.
\end{prop}
\begin{proof}
Let $t=\varphi_r(T)$. By Proposition \ref{poutineroidumonde}, it suffices to prove that every minimal b-$t$-atom is feasible to have $\varphi(T)=\varphi_r(T)$. Let $T'$ be a minimal b-$t$-atom and let $N[T']=V(T')\cup N(T')$. By Lemma \ref{prout}, $T'$ has at most $t$ connected components. 
Let $u$ be a vertex of $N(T')$ with a maximal number of neighbors in $N[T']$.
Since $T'$ has at most $t$ connected components and $T$ is a tree, $u$ has at most $t$ neighbors in $N[T']$.

Our proof consists in extending the coloring of $T'$ induced by $D_1$, \ldots, $D_t$ to $N(T')$ using colors from $\{1,\ldots t\}$. For $t=2$, the proof is trivial since the only minimal b-2-atom is $P_2$ and we can easily extend the coloring to $N(P_2)$. Thus we can suppose that $t\ge 3$.
If $u$ has at most $t-1$ neighbors in $N[T']$, then we can extend the coloring. Thus, we suppose that $u$ has $t$ neighbors in $N[T']$. In this case, $T'$ has $t$ connected components which are all stars.
Each vertex of $N(u)\cap N[T']$ is either a vertex of a connected component of $T'$ or a vertex in $N(T')$ which is adjacent to one vertex of $V(T')$. In these two cases the vertices of $N(u)\cap N[T']$ should be in or be adjacent to vertices of disjoint connected components of $T'$.
Thus the vertices of $N(u)\cap N(T')$ have at most two neighbors in $N[T']$: the vertex $u$ and another vertex of $T'$ (otherwise, there is a cycle in $T$) . We begin by giving a color from $\{1,\ldots,t\}$ to the vertices of $N(T')\setminus \{u\}$.
The vertex $u$ can not be adjacent to all vertices of $C(T')$ since otherwise it would contradict $t=\varphi_r(T)$. Let $v\in N[T']\setminus C(T')$ be a neighbor of $u$. If $v\in N(T')$, then $v$ has at most two neighbors in $N[T']$ and $v$ can be recolored in order to color $u$.
If all neighbors of $u$ are in $T'$, then $v\in N(c_i)$, for $i\in\{1,\ldots, t\}$ and we can exchange the color of $v$ with the color of a vertex $w\in N(c_i)\setminus \{v\}$ in order to color $u$ (since $t\ge 3$, $N(c_i)\setminus \{v\}$ is not empty). Finally, the vertices of $N(w)\cap N(T')$ can be recolored if we have obtained an improper coloring by recoloring $w$.
\end{proof}
The \emph{girth} of a graph $G$ is the length of a smallest cycle in $G$.
We finish this paper by proving that when a graph $G$ has sufficiently large girth, we have $\varphi(G)=\varphi_r(G)$, thus extending Proposition \ref{tree44}.
\begin{theo}
Let $G$ be a graph with girth $g$ and $\varphi_r(G)\ge 3$. If $g\ge 7$, then $\varphi(G)=\varphi_r(G)$.
\end{theo}
\begin{proof}
Let $t=\varphi_r(G)$.
By Proposition \ref{poutineroidumonde}, it suffices to prove that every minimal b-$t$-atom is feasible to have $\varphi(G)=\varphi_r(G)$. Let $A_t$ be a minimal b-$t$-atom. 
Our proof consists in extending the coloring of $A_t$ induced by $D_1$, \ldots, $D_t$ to $N(A_t)$ using colors from $\{1,\ldots t\}$. Thus, we consider that the vertices of $A_t$ are already colored.

For a vertex $u\in N(A_t)$, we denote by $I_c(u)$ the set $\{ i\in\{1,\ldots,t\} |\ \exists v\in N(u)\cap N[c_i]\}$. For a vertex $u\in V(A_t)$, we denote by $c^u$ a neighbor of $u$ in $C(A_t)$ if $u\notin C(A_t)$ or the vertex $u$ itself if $u\in C(A_t)$.
Finally, we denote by $N[A_t]$, the set of vertices $V(A_t)\cup N(A_t)$.
In the different cases, when we describe a cycle of length at most $k$ by $u_1$-$\ldots$-$u_k$, it is assumed that, depending the configuration, consecutive symbols can denote the same vertex.
In this proof, any considered vertex is supposed to be in $N[A_t]$.
We begin by proving the following properties:
\begin{enumerate}
\item[i)] No vertex of $N(A_t)$ is adjacent to two vertices of $N[c_i]$, for $1\le i\le t$;
\item[ii)] If $u,v\in N(A_t)$ and $i\in I_c(u)\cap I_c(v)$, then $u$ and $v$ are not adjacent and have no common neighbor in $N(A_t)-c_i$;
\item[iii)] If $u,v \in N[c_i]$ and $u',v'\in N[c_j]$, $u\neq v$, $u'\neq v'$, for some $i$ and $j$, $1\le i < j \le t$, then the subgraph induced by $\{u,v,u',v'\}$ contains at most one edge.
\end{enumerate}
\begin{description}
\item[i)] If $u$ is adjacent to two vertices of $N[c_i]$, for some $i$, $1\le i\le t$, then $u$ is in a cycle of length at most $4$. This cycle contains $u$, $c_i$ and one or two vertices of $N[c_i]$.
\item[ii)] If $u$ and $v$ are adjacent or have a common neighbor, then $u$ and $v$ belong to a cycle of length at most $6$. This cycle contains $u$, $v$, vertices of $N[c_i]$ and possibly the common neighbor of $u$ and $v$ in $N(A_t)-c_i$, for $i$ an integer such that $i\in I_c(u)\cap I_c(v)$.
\item[iii)] If the subgraph induced by $\{u,v,u',v'\}$ contains at least $2$ edges, then there is a cycle of length at most $6$ in $G$. This cycle is $u$-$v$-$c_i$ if $u$ and $v$ are adjacent, $u'$-$v'$-$c_j$ if $u'$ and $v'$ are adjacent or the cycle $u$-$c_i$-$v$-$u'$-$v'$-$c_j$, otherwise.
\end{description}

We are going to prove that either each vertex $u\in N(A_t)$ can be colored with colors from $\{1,\ldots,t\}$ or the graph $G$ contains a b-$(t+1)$-atom (which contradicts $\varphi_r (G)=t$).
By properties i) and ii), any vertex of $N(A_t)$ has at most $t$ neighbors in $N[A_t]$. 
Hence we may suppose that any vertex $u\in N(A_t)$ with less than $t$ neighbors in $N[A_t]$ is already colored and only consider vertices of $N(A_t)$ with $t$ neighbors in $N[A_t]$.
For a vertex $u\in N[A_t]$, a color $i$ is said to be \emph{available} for $u$ if no vertex has color $i$ in $N(u)\cap N[A_t]$ (and therefore, $u$ has no available color if the colors $1,\ldots,t$ are not available for $u$). Let $N_{*}(A_t)$ be the set of vertices in $N(A_t)$ with no available colors.

We define the following three sets:
\begin{itemize}
\item $N_1= \{u\in N(A_t) |\ N(u)\cap (V(A_t)\setminus C(A_t))\neq \emptyset ,\ N(u)\cap N(A_t)=\emptyset \}$;
\item $N_2= \{u\in N(A_t) |\ N(u)\cap (V(A_t)\setminus C(A_t))\neq \emptyset ,\ N(u)\cap N(A_t)\neq \emptyset \}$;
\item $N_3=\{u\in N(A_t) |\  N(u)\cap (V(A_t)\setminus C(A_t))= \emptyset \}$.
\end{itemize}

We can remark that $N_1\cup N_2\cup N_3=N(A_t)$.

In the remainder of the proof we will first consider the vertices of $N_1$; secondly the vertices of $N_2$; and finally the vertices of $N_3$.

\begin{description}
\item[Case 1:] vertices of $N_1$.\\
Let $u$ be a vertex of $N_1$.
We recall that, by the above assumption, $u$ has exactly $t$ neighbors in $A_t$. 
Moreover, by Property i), $|I_c(u)|=t$.
Let $c_i\in C(A_t)$. We denote by $A_{*}^{i}$ the vertices of $N(c_i)$ which have a neighbor in $N_{*}(A_t)$.
Notice that a vertex $v\in A_{*}^{i}$ can not have a neighbor $x$ in $V(A_t)\setminus \{c_i\}$ since otherwise it would create a cycle $v$-$x$-$c^x$-$v'$-$u$, for $u$ the neighbor of $v$ in $N_1\cap N_{*}(A_t)$ and $v'$ the neighbor of $u$ in $N[c^x]$. This cycle has length at most $5$, contradicting $g\ge 7$.
If for a vertex $c_i\in C(A_t)$ we have $|A_{*}^{i}|\ge 2$, we exchange the colors of the vertices of $A_{*}^{i}$ by doing a cyclic permutation of their colors. Afterwards, we obtain that some vertices of $N_1\cap N_{*}(A_t)$ have now an available color and we recolor them by any available color.
Finally, we color the vertices of $N_1$, when possible, by any available color.
Let $N_{**}(A_t)$ be the set of the remaining uncolored vertices of $N_1$.
In the following subcases, we recolor at most once the vertices of $N[c_i]$, for $i\in\{1,\ldots t\}$, since any two vertices of $N_{**}(A_t)$ can not both have neighbors in $N(c_i)$.

By considering that $N_{**}(A_t)\neq \emptyset$ (or else we have nothing more to do in Case 1)), we can suppose that for every two integers $i$, $j$, $1\le i\neq j\le t$, we have $N[c_i]\cap N[c_j]= \emptyset$. Otherwise, if there exists a vertex $u\in N_{**}(A_t)$ and a vertex $w\in N[c_i]\cap N[c_j]$, there is a cycle $u$-$v$-$c_i$-$w$-$c_j$-$v'$ of length at most $6$, for $v$ a neighbor of $u$ in $N[c_i]$ and $v'$ a neighbor of $u$ in $N[c_j]$. Thus, we obtain that if $N_{**}(A_t)\neq \emptyset$, then every vertex $c_i\in C(A_t)$ has only one neighbor of color $j$, for $1\le i\neq j\le t$, since otherwise it would contradict the minimality of $A_t$ (by removing one vertex of color $j$).

We then consider the two following subcases, for $u\in N_{**}(A_t)$.
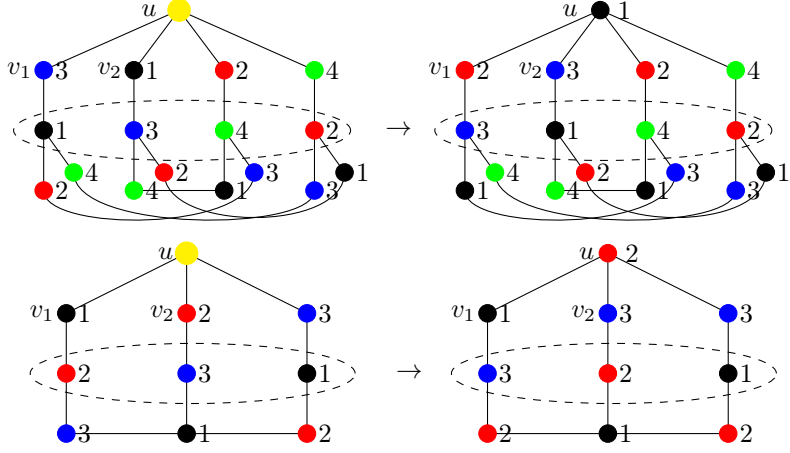
\begin{figure}[t]
\begin{center}
\begin{tikzpicture}[scale=0.8]

\draw (0,0) -- (0,2);
\draw (0,1) -- (0.5,0.3);
\draw (1.5,0) -- (1.5,2);
\draw (1.5,1) -- (2,0.3);
\draw (3,0) -- (3,2);
\draw (3,1) -- (3.5,0.3);
\draw (4.5,0) -- (4.5,2);
\draw (4.5,1) -- (5,0.3);
\draw(2.25,3) -- (0,2);
\draw(2.25,3) -- (1.5,2);
\draw(2.25,3) -- (3,2);
\draw(2.25,3) -- (4.5,2);
\draw(0,0) .. controls (0,-0.7) and (3.5,-0.7).. (3.5,0.3);
\draw(0.5,0.3) .. controls (0.5,-0.7) and (4.5,-0.7).. (4.5,0);
\draw(1.5,0) -- (3,0);
\draw(2,0.3) .. controls (2,-0.7) and (5,-0.7).. (5,0.3);

\node at (0,0) [circle,draw=red,fill=red, scale=0.7] {};
\node at (0,2) [circle,draw=blue,fill=blue, scale=0.7] {};
\node at (0,1) [circle,draw=black,fill=black, scale=0.7] {};
\node at (0.5,0.3) [circle,draw=green,fill=green, scale=0.7] {};
\node at (2,0.3) [circle,draw=red,fill=red, scale=0.7] {};
\node at (1.5,0) [circle,draw=green,fill=green, scale=0.7] {};
\node at (1.5,1) [circle,draw=blue,fill=blue, scale=0.7] {};
\node at (1.5,2) [circle,draw=black,fill=black, scale=0.7] {};
\node at (3,2) [circle,draw=red,fill=red, scale=0.7] {};
\node at (3,1) [circle,draw=green,fill=green, scale=0.7] {};
\node at (3,0) [circle,draw=black,fill=black, scale=0.7] {};
\node at (3.5,0.3) [circle,draw=blue,fill=blue, scale=0.7] {};
\node at (2.25,3) [circle,draw=yellow,fill=yellow, scale=0.9] {};
\node at (4.5,2) [circle,draw=green,fill=green, scale=0.7] {};
\node at (4.5,1) [circle,draw=red,fill=red, scale=0.7] {};
\node at (4.5,0) [circle,draw=blue,fill=blue, scale=0.7] {};
\node at (5,0.3) [circle,draw=black,fill=black, scale=0.7] {};
\node at (0.3,0) {2};
\node at (0.3,1) {1};
\node at (0.3,2) {3};
\node at (0.8,0.3) {4};
\node at (1.8,2) {1};
\node at (1.8,1) {3};
\node at (1.8,0) {4};
\node at (2.3,0.3) {2};
\node at (3.3,0) {1};
\node at (3.3,1) {4};
\node at (3.3,2) {2};
\node at (3.8,0.3) {3};
\node at (4.8,2) {4};
\node at (4.8,1) {2};
\node at (4.8,0) {3};
\node at (5.3,0.3) {1};
\draw[dashed] (2.3,1) ellipse (2.8cm and 0.5cm);
\node at (1.75,3) {$u$};
\node at (-0.4,2) {$v_1$};
\node at (1.1,2) {$v_2$};
\node at (5.9,1) {$\rightarrow$};
\draw (0+7,0) -- (0+7,2);
\draw (0+7,1) -- (0.5+7,0.3);
\draw (1.5+7,0) -- (1.5+7,2);
\draw (1.5+7,1) -- (2+7,0.3);
\draw (3+7,0) -- (3+7,2);
\draw (3+7,1) -- (3.5+7,0.3);
\draw (4.5+7,0) -- (4.5+7,2);
\draw (4.5+7,1) -- (5+7,0.3);
\draw(2.25+7,3) -- (0+7,2);
\draw(2.25+7,3) -- (1.5+7,2);
\draw(2.25+7,3) -- (3+7,2);
\draw(2.25+7,3) -- (4.5+7,2);
\draw(0+7,0) .. controls (0+7,-0.7) and (3.5+7,-0.7).. (3.5+7,0.3);
\draw(0.5+7,0.3) .. controls (0.5+7,-0.7) and (4.5+7,-0.7).. (4.5+7,0);
\draw(1.5+7,0) -- (3+7,0);
\draw(2+7,0.3) .. controls (2+7,-0.7) and (5+7,-0.7).. (5+7,0.3);

\node at (0+7,2) [circle,draw=red,fill=red, scale=0.7] {};
\node at (0+7,1) [circle,draw=blue,fill=blue, scale=0.7] {};
\node at (0+7,0) [circle,draw=black,fill=black, scale=0.7] {};
\node at (0+7.5,0.3) [circle,draw=green,fill=green, scale=0.7] {};
\node at (2+7,0.3) [circle,draw=red,fill=red, scale=0.7] {};
\node at (1.5+7,0) [circle,draw=green,fill=green, scale=0.7] {};
\node at (1.5+7,2) [circle,draw=blue,fill=blue, scale=0.7] {};
\node at (1.5+7,1) [circle,draw=black,fill=black, scale=0.7] {};
\node at (3+7,2) [circle,draw=red,fill=red, scale=0.7] {};
\node at (3+7,1) [circle,draw=green,fill=green, scale=0.7] {};
\node at (3+7,0) [circle,draw=black,fill=black, scale=0.7] {};
\node at (3.5+7,0.3) [circle,draw=blue,fill=blue, scale=0.7] {};
\node at (2.25+7,3) [circle,draw=black,fill=black, scale=0.7] {};
\node at (4.5+7,2) [circle,draw=green,fill=green, scale=0.7] {};
\node at (4.5+7,1) [circle,draw=red,fill=red, scale=0.7] {};
\node at (4.5+7,0) [circle,draw=blue,fill=blue, scale=0.7] {};
\node at (5+7,0.3) [circle,draw=black,fill=black, scale=0.7] {};
\node at (0.3+7,2) {2};
\node at (0.3+7,0) {1};
\node at (0.3+7,1) {3};
\node at (0.8+7,0.3) {4};
\node at (1.8+7,1) {1};
\node at (1.8+7,2) {3};
\node at (1.8+7,0) {4};
\node at (2.3+7,0.3) {2};
\node at (3.3+7,0) {1};
\node at (3.3+7,1) {4};
\node at (3.3+7,2) {2};
\node at (3.8+7,0.3) {3};
\node at (4.8+7,2) {4};
\node at (4.8+7,1) {2};
\node at (4.8+7,0) {3};
\node at (5.3+7,0.3) {1};
\node at (2.65+7,3) {1};
\draw[dashed] (2.3+7,1) ellipse (2.8cm and 0.5cm);
\node at (1.75+7,3) {$u$};
\node at (-0.4+7,2) {$v_1$};
\node at (1.1+7,2) {$v_2$};
\end{tikzpicture}
\begin{tikzpicture}[scale=0.8]
\draw (0,0) -- (0,2);
\draw (2,0) -- (2,2);
\draw (4,0) -- (4,2);
\draw(0,0) -- (4,0);
\draw(2,3) -- (0,2);
\draw(2,3) -- (2,2);
\draw(2,3) -- (4,2);
\node at (0,1) [circle,draw=red,fill=red, scale=0.7] {};
\node at (0,0) [circle,draw=blue,fill=blue, scale=0.7] {};
\node at (0,2) [circle,draw=black,fill=black, scale=0.7] {};
\node at (2,2) [circle,draw=red,fill=red, scale=0.7] {};
\node at (2,1) [circle,draw=blue,fill=blue, scale=0.7] {};
\node at (2,0) [circle,draw=black,fill=black, scale=0.7] {};
\node at (4,0) [circle,draw=red,fill=red, scale=0.7] {};
\node at (4,2) [circle,draw=blue,fill=blue, scale=0.7] {};
\node at (4,1) [circle,draw=black,fill=black, scale=0.7] {};
\node at (2,3) [circle,draw=yellow,fill=yellow, scale=0.9] {};
\node at (0.3,2) {1};
\node at (0.3,0) {3};
\node at (0.3,1) {2};
\node at (2.3,2) {2};
\node at (2.3,0) {1};
\node at (2.3,1) {3};
\node at (4.3,2) {3};
\node at (4.3,0) {2};
\node at (4.3,1) {1};
\draw[dashed] (2.1,1) ellipse (2.7cm and 0.5cm);
\node at (1.65,3) {$u$};
\node at (-0.4,2) {$v_1$};
\node at (1.6,2) {$v_2$};
\node at (5.7,1) {$\rightarrow$};
\draw (0+7,0) -- (0+7,2);
\draw (2+7,0) -- (2+7,2);
\draw (4+7,0) -- (4+7,2);
\draw(0+7,0) -- (4+7,0);
\draw(2+7,3) -- (0+7,2);
\draw(2+7,3) -- (2+7,2);
\draw(2+7,3) -- (4+7,2);
\node at (2+7,1) [circle,draw=red,fill=red, scale=0.7] {};
\node at (2+7,2) [circle,draw=blue,fill=blue, scale=0.7] {};
\node at (0+7,2) [circle,draw=black,fill=black, scale=0.7] {};
\node at (0+7,0) [circle,draw=red,fill=red, scale=0.7] {};
\node at (0+7,1) [circle,draw=blue,fill=blue, scale=0.7] {};
\node at (2+7,0) [circle,draw=black,fill=black, scale=0.7] {};
\node at (4+7,0) [circle,draw=red,fill=red, scale=0.7] {};
\node at (4+7,2) [circle,draw=blue,fill=blue, scale=0.7] {};
\node at (4+7,1) [circle,draw=black,fill=black, scale=0.7] {};
\node at (2+7,3) [circle,draw=red,fill=red, scale=0.7] {};
\node at (0.3+7,2) {1};
\node at (2.3+7,2) {3};
\node at (2.3+7,1) {2};
\node at (0.3+7,0) {2};
\node at (2.3+7,0) {1};
\node at (0.3+7,1) {3};
\node at (4.3+7,2) {3};
\node at (4.3+7,0) {2};
\node at (4.3+7,1) {1};
\node at (2.4+7,3) {2};
\draw[dashed] (2.1+7,1) ellipse (2.7cm and 0.5cm);
\node at (1.65+7,3) {$u$};
\node at (-0.4+7,2) {$v_1$};
\node at (1.6+7,2) {$v_2$};
\end{tikzpicture}
\end{center}
\caption{Possible configurations in Subcases 1.2.1 (on the top) and 1.2.2 (on the bottom) before (on the left) and after (on the right) the recoloring process.}
\label{lastproof}
\end{figure}

\item[Subcase 1.1:] $u$ has exactly one neighbor in $V(A_t)\setminus C(A_t)$.\\
Let $v'$ be the neighbor of $u$ in $V(A_t)\setminus C(A_t)$ and let $c'$ be the color of $v'$.
Notice that no vertex $x$ from $N[c^{v'}]$ has a neighbor $y$ in $V(A_t)\setminus N[c^{v'}]$, since otherwise it would create a cycle $u$-$v'$-$c^{v'}$-$x$-$y$-$c^y$ of length at most $6$. Consequently, we can exchange the color of $v'$ with the color of one vertex from $N(c^{v'})$ and color $u$ by $c'$.

\item[Subcase 1.2:] $u$ has more than one neighbor in $V(A_t)\setminus C(A_t)$.\\
Let $v_1$ and $v_2$ be two neighbors of $u$ in $V(A_t)\setminus C(A_t)$.
Let $c'$ be the color of $v_1$ and let $c''$ be the color of $v_2$.

If $v_1$ has a neighbor $x\in V(A_t)\setminus N[c^{v_1}]$, then there exists a cycle $u$-$v_1$-$x$-$c^x$-$v'$ in $G$, with $v'$ a neighbor of $c^x$ in $N(u)$ (in the case $c^x$ is not a neighbor of $u$). Similarly if $v_2$ has a neighbor in $V(A_t)\setminus N[c^{v_2}]$, then there is a cycle of length at most $5$ in $G$.
Consequently, we can suppose that $v_1$ has no neighbor in $V(A_t)\setminus N[c^{v_1}]$ and that $v_2$ has no neighbor in $V(A_t)\setminus N[c^{v_2}]$. If there exists a vertex of $N(c^{v_1}) \setminus \{v_1\}$ with no neighbor of color $c'$, then we exchange the color of $v_1$ with the color of this vertex and color $u$ by $c'$. If there exists a vertex of $N(c^{v_2})\setminus \{v_2\}$ with no neighbor of color $c''$, then we exchange the color of $v_2$ with the color of this vertex and color $u$ by $c''$. 
Thus, we may suppose that every vertex $w$ of $N(c^{v_1})\setminus \{v_1\}$ (of $N(c^{v_2})\setminus \{v_2\}$, respectively) has a neighbor $\overline{w}$ of color $c'$ ($c''$, respectively) in $V(A_t)$.
We consider three subscases in order to color to $u$.

\item[Subcase 1.2.1:] the vertices $v_1$ and $c^{v_2}$ have the same color and the vertices $v_2$ and $c^{v_1}$ have the same color. \\
Notice that no vertex $w\in N(c^{v_1})$ is adjacent to $c^{v_2}$ since otherwise $u$-$v_1$-$c^{v_1}$-$w$-$c^{v_2}$-$v_2$ would be a cycle of length at most $6$ in $G$. For the same reason, no vertex $w\in N(c^{v_2})$ is adjacent to $c^{v_1}$.
Thus, by Property iii), no vertex $w \in N[c^{v_1}]\cup N[c^{v_2}]$ has a neighbor $x\in V(A_t)\setminus (N[c^{v_1}]\cup N[c^{v_2}]\cup \{\overline{w}\})$, since there exists a vertex $y\in N(c^{w})$ with neighbor $\overline{y}\in N(c^{x})$.
There could exist two adjacent vertices $w$ and $w'$ with $w\in N(c^{v_1})$ and $w'\in N(c^{v_2})$. 
However, the vertex $w'$ has no neighbor of color $c''$ in $A_t$ since $w'$ and $v_2$ can not be adjacent and there does not exist a second vertex of color $c''$ in $N(c^{v_2})$.
Consequently, we can exchange the color of $v_1$ with the color of $v_2$, the color of $c^{v_1}$ with the color of $c^{v_2}$ and afterward we can exchange the color of one vertex from $N(c^{v_1})\setminus \{v_1\}$ with the color of $v_1$ and color $u$ by $c''$. The top of Figure~\ref{lastproof} illustrates this recoloring process on a minimal b-$4$-atom fulfilling the hypothesis of Subcase 1.2.1.

\item[Subcase 1.2.2:] the vertices $v_1$ and $c^{v_2}$ do not have the same color.\\
Let $i$ be the color of $c^{v_1}$ and $j$ be the color of $c^{v_2}$.
In this case, we exchange the color of $c^{v_1}$ with the color of $c^{v_2}$ and the color of the vertex $w$ of color $j$ in $N(c^{v_1})$ with the color of the vertex $w'$ of color $i$ in $N(c^{v_2})$. For this, we have to suppose that $w$ is not adjacent to a vertex of color $i$ and that $w'$ is not adjacent to a vertex of color $j$. For $t\ge 4$, such vertices $w$ and $w'$ exist since at most one vertex of $N(c^{v_1})$ has a neighbor of color $j$ (otherwise, it would contradict Property iii) since every vertex of $N(c^{v_1})\setminus \{v_1\}$ has already a neighbor in $V(A_t)$ of color $c'$) and at most one vertex of $N(c^{v_2})$ has a neighbor of color $i$. If $t=3$, then the only (up to isomorphism) b-$3$-atom with a coloring fulfilling all these hypothesis (up to color permutation) is illustrated at the bottom of Figure~\ref{lastproof}, along with the recoloring process. In this b-$3$-atom, no more edge can be added (otherwise, it would create a cycle of length at most 6).

\item[Subcase 1.2.3:] the vertices $v_2$ and $c^{v_1}$ do not have the same color.\\
We proceed as for the previous subcase by considering $v_2$ instead of $v_1$ and $c^{v_1}$ instead of $c^{v_2}$.

\item[Case 2:] vertices of $N_2$.\\
Since each pair of adjacent vertices $u,v\in N(A_t)$ satisfies Property ii), we obtain that $I_c(u)\cap I_c(v)= \emptyset$.
We color each vertex $u\in N_2$ by a color $i\in I_c(u)$ such that $u$ and $c_i$ are not adjacent.

\item[Case 3:] vertices of $N_3$.\\
Notice that, by definition, a vertex of $C(A_t)$ has no available color.
Let $u\in N_3$.
We begin by coloring $u$ with any available color if it has some. If $u$ has no available color, there could exist a color $i$ such that every vertex of $N(u)$ with color $i$ has an available color (these vertices should be in $N(A_t)$). If such color $i$ exists, we recolor these vertices of color $i$ by any available color and give color $i$ to $u$. If such color $i$ does not exist, then the set of vertices at distance at most $2$ from $u$ induces a b-$(t+1)$-atom with center $N[u]$. It can be noted that the recolored vertices are in $N(A_t)$ since $N(u)\cap V(A_t)\subseteq C(A_t)$.
\end{description}
We finish this proof by illustrating that the obtained coloring is a b-$t$-coloring of $N[A_t]$. In case 1, we have modified the coloring of $A_t$. However, since we have exchanged the colors of well-chosen vertices in order that every vertex of $C(A_t)$ still has neighbor of every color from $\{1,\ldots,t\}$ except its own color, this coloring remains a b-$t$-coloring. In case 3, we have only changed the color of vertices from $N(A_t)$. 

\end{proof}
We think that the previous theorem can be useful to determine the family of graphs of girth at least 7 satisfying $\varphi(G)= m(G)$. It has already been proven that graphs of girth at least 7 have b-chromatic number at least $m(G)-1$ \cite{CAM2013}.
\begin{cor}
Let $G$ be a graph of girth at least $7$ and of order $n$ and let $t$ be an integer.
There exists an algorithm in time $O(n^{t^2})$ to determine if $\varphi(G)\ge t$.
\end{cor}
\section{Open questions}
We conclude this article by listing few open questions.
\begin{enumerate}
\item For which family of graphs are the b-relaxed number and the b-chromatic number equal?
\item Does there exist an easy characterization of feasible b-$t$-atoms?
\item Does there exist an FPT algorithm, with parameter $t$, to determine if $\varphi(G)\ge t$?
\end{enumerate}
\bibliographystyle{abbrv}
\bibliography{bib}
\end{document}